\documentclass[12pt, a4paper]{amsart}

\pdfoutput=1

\usepackage[T1]{fontenc}

\usepackage{euscript,amsfonts,amssymb,amsmath,amscd,amsthm,enumerate,hyperref}

\usepackage{comment}

\usepackage[font={small}]{caption}

\usepackage{mathrsfs}

\usepackage{float}

\usepackage{graphicx}
\usepackage[left=1in,right=1in,top=1in,bottom=1in]{geometry}

\usepackage{mathtools}
\mathtoolsset{showonlyrefs}

\usepackage{bbm}

\newcommand{\E}{\mathbb E}
\newcommand{\EE}{\mathbb E}

\newcommand{\ee}{\mathcal{E}}
\newcommand{\dd}{\mathrm{d}}
\newcommand{\pp}{\mathbb{P}}
\newcommand{\qq}{\mathbb{Q}}
\newcommand{\rr}{\mathbb{R}}

\newcommand{\eq}{\begin{equation}}
\newcommand{\en}{\end{equation}}

\newcommand{\bone}{\mathbf{1}}

\newcommand{\F}{\mathcal{F}}
\newcommand{\G}{\mathcal{G}}
\newcommand{\A}{\mathcal{A}}
\newcommand{\W}{W^{\perp}}

\newtheorem{theorem}{Theorem}[section]
\newtheorem*{theorem*}{Theorem}
\newtheorem{lemma}[theorem]{Lemma}
\newtheorem{definition}[theorem]{Definition}
\newtheorem{proposition}[theorem]{Proposition}

\newtheorem{corollary}[theorem]{Corollary}
\newtheorem{assumption}[theorem]{Assumption}

\theoremstyle{remark}
\newtheorem{rmk}[theorem]{Remark}

\title[Power Mixture FPPs]{Power Mixture Forward Performance Processes}

\author{Levon Avanesyan}
\address{ORFE Department, Princeton University, Princeton, NJ 08544, USA}
\email{levon.avanesyan23@gmail.com}
\thanks{Princeton University}

\author{Ronnie Sircar}
\address{ORFE Department, Princeton University, Princeton, NJ 08544, USA}
\email{sircar@princeton.edu}

\date{\today}

\usepackage{amsopn}

\ifpdf
\hypersetup{
  pdftitle={Power Mixture Forward Performance Processes},
  pdfauthor={L. Avanesyan, and R. Sircar}
}
\fi

\begin{document}

\begin{abstract}
We consider the forward investment problem in market models where the stock prices are continuous semimartingales adapted to a Brownian filtration. We construct a broad class of forward performance processes with initial conditions of power mixture type, $u(x) = \int_{\mathbb{I}} \frac{x^{1-\gamma}}{1-\gamma }\nu(\dd \gamma)$.  We proceed to define and fully characterize two-power mixture forward performance processes with constant risk aversion coefficients in the interval $(0,1)$, and derive properties of two-power mixture forward performance processes when the risk aversion coefficients are continuous stochastic processes. Finally, we discuss the problem of managing an investment pool of two investors, whose respective preferences evolve as power forward performance processes.\\


\smallskip
\noindent \textbf{AMS Subject Classificaton:} 91B16, 60G44, 91G10
\end{abstract}

\maketitle

\section{Introduction}
Consider an investor with initial capital $X_0 = x > 0$ aiming to invest in a market consisting of a riskless bank account with zero interest rate and $n \geq 1$ stocks whose price processes $S^1,\, S^2, \, \ldots, \, S^n$ follow the continuous semimartingale dynamics 
\begin{align}\label{eq:stock.dynamics}
    \frac{\dd S_t^i}{S_t^i} = \mu_t^i\dd t + \sum_{j = 1}^{d_W} \sigma_t^{ji}\dd W_t^j, \quad i = 1,\, 2,\, \ldots,\, n. 
\end{align}
Here, the stochastic processes $\mu = (\mu^1, \mu^2, \ldots, \mu^n)$ and $\sigma = (\sigma^{ji})_{j,i=1}^{d_W, n}$, for some $d_W \geq 1$, are continuous and adapted to the filtration $\mathbb{F} = (\F_t)_{t\geq 0}$ generated by a pair of $d_W$ and $d_{\W}$-dimensional standard Brownian motions $(W, \W)$, for some $d_{\W} \geq 0$, on a filtered probability space $(\Omega, \F, \mathbb{F}, \pp)$. When $d_{\W}=0$ the market is complete, and incomplete otherwise.

Given an initial investment capital $x$, and a choice of a self-financing portfolio with allocations $\pi_t = (\pi_t^1, \pi_t^2, \ldots,\pi_t^n)$ in units of fraction of current wealth among the $n$ stocks, the investor's wealth generated from holding this portfolio will have the dynamics
\begin{align}\label{eq:wealth.process}
    \frac{\dd X_t^{\pi}}{X_t^{\pi}} &= (\sigma_t\pi_t)^{\top} \lambda_t \dd t +(\sigma_t \pi_t)^{\top} \dd W_t, \quad X_0^{\pi} = x, 
\end{align}
with $\lambda_t=(\sigma_t^{-1})^\top\mu_t$ denoting the Sharpe ratio, $\sigma_t^{-1}$ the Moore-Penrose inverse of $\sigma_t$, and $\top$ denoting transpose. 
\begin{assumption}\label{assump.var}
	The process  $\sup_{t \in [0,T]}\|\lambda_t\|$
	is bounded almost surely for all $T>0$.
\end{assumption}
The choice of an optimal portfolio for the investor is determined by the admissible portfolio set and the investor's personalized investment performance criterion.

\begin{definition}\label{def.admissibility set}
Consider an $\F_t$-progressively measurable portfolio process $\pi_t$.
\begin{enumerate}[i.]
    \item We will say that the portfolio process is $T$-admissible, $\pi \in \A_T$, if for all $t\leq T$ $X_t^{\pi} \geq 0$, $\int_0^t |X_s^{\pi}\pi_s^{\top}\mu_s| \dd s < \infty$, and $\int_0^t |X_s^{\pi} \sigma_s \pi_s |^2 \dd s < \infty$. 
    \item The portfolio process is forward admissible, $\pi \in \A$, if it is $T$-admissible for all $T> 0$. That is $\A = \underset{T>0}{\cap} \A_T$.
\end{enumerate}
\end{definition}
\begin{definition}
Consider an $\F_t$-progressively measurable portfolio process $\pi_t$, a compact interval $\mathbb{I}$, a positive measure $\nu(\cdot)$, and $v \in [1,\infty)$.
\begin{enumerate}[i.]
    \item We will say that the portfolio process is $T$-$v$-admissible, $\pi \in \A_T^{v}$, if $\pi \in \A_T$,
    \begin{align}\label{eq:admiss.integral}
        \int_0^T\int_{\mathbb{I}}\E\big[(X_t^{\pi})^{2v(1-\gamma)}\|\sigma_t\pi_t\|^{2v}\big]\nu(\dd \gamma) \dd t < \infty,
    \end{align}
    and for some $u>1$
    \begin{align}\label{eq:admiss.sup}
        \sup_{t\in[0,T]}\int_{\mathbb{I}}\E\big[(X_t^{\pi})^{2uv(1-\gamma)}\big]\nu(\dd \gamma) < \infty.
    \end{align}
    \item The portfolio process is $v$-forward admissible, $\pi \in \A^{v}$, if it is $T$-admissible for all $T> 0$. That is $\A^{v} = \underset{T>0}{\cap} \A_T^v$.
\end{enumerate}
\end{definition}

In this paper we study the optimal investment problem under forward investment criteria, originally introduced and developed in \cite{musiela2006investments}, \cite{musiela2007investment}, as well as in \cite{HENDERSON20071621}. The forward problem addresses investment with an a priori unknown time horizon over which the investor's utility function may evolve. The forward investment problem then is to find an $\mathbb{F}$-progressively measurable process $U_{\cdot}(\cdot)\,:\, [0,\infty) \times (0,\infty) \to \rr$ and a $\pi^* \in \A$ satisfying
\begin{enumerate}[(a)]
    \item with probability one, all functions $x \mapsto U_t,\, t\geq 0$ are strictly concave and increasing;
    \item for each $\pi \in \A$, the process $U_t(X_t^{\pi}),\, t\geq 0$ is an $\mathbb{F}$-supermartingale;
    \item the process $U_t(X_t^{\pi^*}),\, t\geq 0$ is an $\mathbb{F}$-martingale.
\end{enumerate}

The process $U_{\cdot}(\cdot)$ is referred to as a forward performance process (FPP) and its fixed-time projections $U_t(\cdot),\, t\geq 0 $ should be thought of as the (random) utility functions of an investor who is reacting to the information flow $\mathbb{F}$. When conditions (b) and (c) hold locally the process $U_{\cdot}(\cdot)$ is referred to as a ``local FPP'', and when the conditions hold in a true sense the process $U_{\cdot}(\cdot)$ will be referred to as a ``true FPP''. Condition (c) then characterizes the optimal allocations $\pi^* \in \A$ for such an investor.

If $U$ is an It\^{o} process in $t$ and twice differentiable in the wealth parameter, then in \cite{musiela2010stochastic} it was shown that $U$ is a local FPP if and only if it solves the following stochastic partial differential equation
\begin{align}\label{eq.FPP.SPDE}
    \dd U_t(x) = \frac{1}{2} \frac{| \partial_x U_t(x) (\sigma_t^{\top})^{-1}\mu_t + \sigma_t^{\top} (\sigma_t^{\top})^{-1} \partial_x a^W_t(x)|^2}{\partial_{xx}^2 U_t(x)} \dd t + a_t(x)\cdot \dd (W_t, W_t^{\perp}),
\end{align}
where the FPP is characterized by the forward volatility process $a = (a^W,\, a^{W^{\perp}})$. Not only it is hard to solve this SPDE, but also we do not even know what initial conditions would yield existence of a solution. For general initial conditions, a characterization of local FPPs through a solution of a non-linear PDE was first given in \cite{musiela2010portfolio2}. A lot of literature was then dedicated to studying FPPs with the power utility of wealth initial condition
\begin{align}
    U_0(x) = C_0\frac{x^{1-\gamma}}{1-\gamma},
\end{align}
where $C_0 > 0$ and $\gamma \in (0,\infty)/\{1\}$. For Markovian factor market models explicit classes of power local FPPs were constructed in \cite{nadtochiy2014class}, \cite{nadtochiy2015optimal}, \cite{liang2017representation}, and \cite{avanesyan2020construction}. Asymptotic results have been developed in \cite{shkolnikov2015asymptotic}. In recent years, non-Markovian factor market models have gained traction in the literature. Most importantly, in \cite{gatheral2018volatility} it has been demonstrated that a non-Markovian factor market model, where the stochastic volatility is driven by a fractional Brownian motion, fits the observed financial time series quite well. Moreover the said factor is not a semimartingale, but is adapted to a Brownian filtration. Hence, the market dynamics in \eqref{eq:stock.dynamics} include the case of the so-called rough fractional stochastic volatility. In this general setting power FPPs have been fully characterized in \cite{Choulli_2017} and \cite{bo2018forward}, without and with portfolio constraints respectively.

In general, investors do not have CRRA preferences. And even if individual investors do have constant relative risk aversion coefficients, portfolio managers or funds have to satisfy multiple different investors with different risk preferences at the same time. On a more personal level, in the life of a couple there exist major investment decisions that can only be made by pooling the resources.  By maximizing an objective that is a convex combination of different power utilities one can hope to achieve reasonable performance for all types of investors or partners. Hence, we construct FPPs for initial utilities that are formed through convex combinations of utility functions of power form 
$$U_0(x) = \int_{\mathbb{I}} \frac{x^{1-\gamma}}{1-\gamma} \nu (\dd \gamma),$$
where $\mathbb{I}\subset (0,\infty)/ \{ 1 \}$ is the compact interval of the risk aversion coefficients in question. One can think of $\mathbb{I}$ as a pool of investors with different risk aversion parameters, whereby the measure $\nu(\cdot)$ assigns the relative weight of particular risk aversions in the investor pool. In this context, when $\nu(\cdot)$ is a point mass we are dealing with a single rational investor with CRRA preferences, and hence the respective FPP is of power form. However, when $\nu(\cdot)$ is a two-point mass, we would be dealing with an investment pool consisting of two such investing partners, and the respective FPP is what we will call of two-power mixture form. A natural example of such investment pools are joint investment accounts or major financial decisions made by couples. Thus, in Section \ref{sec:two-power-mixture} we extensively characterize the two-power mixture FPPs, and also analyze the case when the joint utility of the couple is not an FPP.

Mixtures of utilities were first introduced in portfolio optimization literature in \cite{fouque2017portfolio} as a sum of two CRRA utilities with differing risk aversion coefficients, which corresponds to the case when $\nu(\cdot)$ is a discrete point-mass measure. Some asymptotic results are derived for the Merton problem, however no results have been derived for FPPs with initial conditions of mixture type. The only work that we are aware of that deals with consistent utilities with initial conditions of power mixture type is \cite{el2017construction}, where the authors consider the problem of finding a dynamic equilibrium by maximizing the aggregate utility of the economy.  Below for the first time we construct such local FPPs and proceed to derive conditions that ensure that the local FPPs are in fact true FPPs.

The results are presented in the following fashion. In Section \ref{sec:power-mixture} we construct a broad class of general power mixture type true FPPs. Using this result, in Section \ref{sec:power_fpp} we obtain a class of power true FPPs, and discuss the meaning of the parametrization. In Section \ref{sec:two-power-mixture} we fully characterize the class of two-power mixture true FPPs when the risk aversion parameters are constant and in the interval $(0,1)$, as well as obtain necessary conditions for the general parametrization. We wrap up Section \ref{sec:two-power-mixture} by discussing the problem of dynamic pooled investment of two rational agents, whose individual preferences evolve as power type FPPs with different risk aversion parameters. In Section \ref{sec:three-power} we construct a three-power mixture FPP, that exhibits the limitations of necessary conditions derived in Section \ref{sec:two-power-mixture}. Finally, in Section \ref{sec:conclusion} we summarize our results and point towards future directions.
\section{General Power Mixture Forward Performances}\label{sec:power-mixture}
Consider an investment pool of investors with CRRA preferences whose risk aversion coefficients belong to a compact interval $\mathbb{I} \subset (0,\infty)/\{1\}$. Let the Borel measure $\nu(\cdot)$ denote the relative weight given to investors based on their risk aversion.   The weighting measure can be chosen based on various criteria, e.g. relative sizes of investment associated with each risk aversion coefficient, investment manager's preferences. The joint utility in the investment pool then is given by
\begin{align}\label{eq:mixture_initial_condition}
U_0(x) = \int_{\mathbb{I}}\frac{x^{1-\gamma}}{1-\gamma} \nu(\dd \gamma).
\end{align}
In this section our aim is to characterize a broad class of FPPs with power mixture initial condition \eqref{eq:mixture_initial_condition}. Thereby, we construct dynamically consistent investment criteria for the above-mentioned investment pools.
\subsection{General characterization of mixture FPPs}
 For a stochastic process $(M_t)_{t\geq 0}$, from here on we will denote its stochastic exponential by $$\ee(M_t) : = \exp\Big(M_t - \frac{1}{2}\langle M \rangle_t\Big).$$
 In our first result we characterize all local FPPs of power-mixture type.
\begin{theorem}\label{thm:local.FPP}
    Suppose the market model \eqref{eq:stock.dynamics}, a compact set $\mathbb{I} \subset (0,\infty)/\{1\}$, and a fixed risk aversion coefficient $\gamma_0 \in \mathbb{I}$. Then, for initial preferences 
    $$U_0(x) = \int_{\mathbb{I}}\frac{x^{1-\gamma}}{1-\gamma} \nu(\dd \gamma),$$
     where $\nu(\cdot)$ is a positive measure, the process
    \begin{align}\label{eq:local.FPP}
        U_t(x) = \int_{\mathbb{I}} \frac{x^{1-\gamma}}{1-\gamma} \ee(M_t^{\gamma}) \ee(V_t^{\gamma})\nu(\dd \gamma)
    \end{align}
    is a local FPP, with an associated optimal portfolio given by a solution to
    \begin{align}\label{eq:optimal.portfolio}
        \sigma_t\pi_t^* &= \frac{1}{\gamma_0}(\lambda_t + H_t^{\gamma_0}),
    \end{align}
    where the pairs $(M_t^{\gamma}, V_t^{\gamma})$ are given by
    \begin{align}\label{eq:local_mtg.gamma}
         M_t^{\gamma} &= \int_{0}^{t} H_s^{\gamma} \cdot \dd W_s + \int_{0}^{t} J_s^{\gamma} \cdot \dd \W_s, \quad H_t^{\gamma} = \frac{\gamma - \gamma_0}{\gamma_0}\lambda_t + \frac{\gamma}{\gamma_0} H_t^{\gamma_0},\\\label{eq:FV.gamma}
    V_t^{\gamma} & = \int_0^t v_s^{\gamma} \dd s,\quad v_t^{\gamma} = -\frac{1-\gamma}{2\gamma} |\lambda_t + H_t^{\gamma} |^2,
    \end{align}
    and $H_t^{\gamma_0}, J_t^{\gamma} \in \F_t$ are such that for all admissible $\pi$, $\pp-$almost surely
    \begin{align}\label{eq:fubini.suff}
    \begin{split}
        \int_{\mathbb{I}} \bigg(\int_0^t &(X_s^{\pi})^{2(1-\gamma)}\big(\ee(M_t^{\gamma})\ee(V_t^{\gamma})\big)^2\\
        &\times\bigg( \bigg|\frac{1}{1-\gamma}H_s^{\gamma} 
        + \sigma_s\pi_s\bigg|^2 + \bigg|\frac{1}{1-\gamma}J_s^{\gamma}\bigg|^2 \bigg) \dd s\bigg)^{\frac{1}{2}}\nu(\dd \gamma) < \infty.
        \end{split}
    \end{align}
\end{theorem}
\begin{rmk}
    The processes $H_t^{\gamma_0}$ and $\{J_t^{\gamma}\}_{\gamma \in \mathbb{I}}$ are not fixed, thereby giving us significant degrees of freedom in constructing FPPs. We discuss interpretations for these processes in the end of this subsection, as well as in Section \ref{subsec.market-view}.
\end{rmk}

\begin{proof}
We will prove this theorem by a verification argument. Fix any $t>0$. Now, for all $\gamma \in \mathbb{I}$ define the processes 
\[U_t^{\gamma}(x) := \frac{x^{1-\gamma}}{1-\gamma} \ee(M_t^{\gamma}) \ee (V_t^{\gamma}),\]
and let us first show that they are indeed local FPPs. For an arbitrary admissible portfolio $\pi$ we have
\begin{align}
    U_t^{\gamma}(X_t^{\pi}) = \frac{1}{1-\gamma}\ee\bigg(\int_0^t (\sigma_s\pi_s)^{\top} \lambda_s\dd s + \int_0^t(\sigma_s\pi_s)^{\top}\dd W_s)\bigg)^{1-\gamma}\ee(M_t^{\gamma}) \ee (V_t^{\gamma}).
\end{align}
Applying It\^{o}'s formula to the process $U_t^{\gamma}(X_t^{\pi})$ we get
\begin{align}
    \frac{\dd U_t^{\gamma}(X_t^{\pi})}{(X_t^{\pi})^{1-\gamma} \ee(M_t^{\gamma}) \ee(V_t^{\gamma})}
    =& \: \frac{1}{1-\gamma}\dd M_t^{\gamma} + (\sigma_t \pi_t)^{\top} \dd W_t  \\
    &+ \bigg(\frac{1}{1-\gamma} v_t^{\gamma} +  (\sigma_t \pi_t)^{\top}( \lambda_t + H_t^{\gamma}) -\frac{\gamma}{2} |\sigma_t \pi_t|^2 \bigg) \dd t \\
    :=&\: D^{\gamma}(\pi_t) \dd t + \frac{1}{1-\gamma}\dd M_t^{\gamma} +  (\sigma_t \pi_t)^{\top} \dd W_t.
\end{align}
Note that $D^{\gamma}(\pi_t)$ is a globally concave function of $\pi_t$, and thus we can find an optimal $\pi_t^{\gamma*}$ by finding a critical point of the function $D^{\gamma}(\cdot)$
\begin{align}
   \partial_{\pi} D^{\gamma}(\pi_t^{\gamma*}) = \sigma_t^{\top}(\lambda_t + H_t^{\gamma}) - \gamma \sigma_t^{\top}\sigma_t \pi_t^{\gamma*},
\end{align}
which yields that a portfolio process solving the equation
\begin{align}
    \sigma_t\pi_t^{\gamma*} = \frac{1}{\gamma}(\lambda_t + H_t^{\gamma})
\end{align}
would maximize the function $D^{\gamma}(\cdot)$. Plugging in the expressions for $H_t^{\gamma}$ we get
\begin{align}
    \sigma_t\pi_t^{\gamma*} = \frac{1}{\gamma_0}(\lambda_t + H_t^{\gamma_0}) = \sigma_t \pi_t^{\gamma_0^*} =: \sigma_t \pi_t^{*}.
\end{align}
Let us now calculate the maximal value of the function $D^{\gamma}(\cdot)$,
\begin{align}
    D^{\gamma}(\pi_t^{*})
   &= \frac{1}{1-\gamma} v_t^{\gamma} + \frac{1}{2\gamma}|\lambda_t + H_t^{\gamma} |^2 = 0,
\end{align}
where the last equality follows from \eqref{eq:local_mtg.gamma}, \eqref{eq:FV.gamma}. Thus, for all $\pi \in \A$ the drift will be non-positive. Thereby, for all $\gamma \in \mathbb{I}$ the processes $U_t^{\gamma}(X_t^{\pi})$ are local supermartingales for all $\pi \in \A$ and local martingales for the portfolio $\pi^{*}$. Note that our initial process can be written as
\begin{align}
    U_t(X_t^{\pi}) =& \: \int_{\mathbb{I}} U_t^{\gamma}(X_t^{\pi}) \nu( \dd \gamma) \\
    =& \: \int_{\mathbb{I}} \int_0^t(X_s^{\pi})^{1-\gamma}\ee(M_s^{\gamma})\ee(V_s^{\gamma}) D^{\gamma}(\pi_s) \dd s \: \nu (\dd \gamma) \\
    & + \sum_{i = 1}^{d_W} \int_{\mathbb{I}} \int_0^t (X_s^{\pi})^{1-\gamma}\ee(M_s^{\gamma})\ee(V_s^{\gamma})\\
    &\times \bigg(\frac{\gamma}{(1-\gamma)\gamma_0} H_s^{\gamma_0} 
    + \frac{\gamma - \gamma_0}{(1-\gamma)\gamma_0}\lambda_s + \sigma_s \pi_s \bigg)_i\dd W_{s,\, i}\, \nu (\dd \gamma) \\ 
    &+ \sum_{j = 1}^{d_{\W}}\int_{\mathbb{I}} \int_0^t \frac{1}{1-\gamma} (X_s^{\pi})^{1-\gamma}\ee(M_s^{\gamma})\ee(V_s^{\gamma}) J_s^{\gamma,\,i}\dd \W_{s,\,j}\, \nu (\dd \gamma).
\end{align}
From \eqref{eq:fubini.suff}, by invoking the stochastic Fubini Theorem (Theorem 2.2) from \cite{veraar2012stochastic}, we get that the last two terms are local martingales. The first summand is always non-positive, since $D^{\gamma}(\pi_s) \leq 0$ for all $\gamma \in \mathbb{I}$, and becomes $0$ when evaluated at the optimal $\pi^*$ (the non-positivity of the integrand also yields the measurability of the integral by Fubini's theorem). Hence, the function $U_t(x)$ takes values in $\rr \cup \{-\infty\}$. For the purposes of optimal portfolio selection $U_{\cdot}(\cdot)$ is well defined, since there exists $\pi^*$ for which its drift is equal to $0$, and therefore the portfolios yielding the $\{-\infty\}$ value for the drift term, and thereby for $U_{\cdot}(\cdot)$, cannot be optimal and therefore can be ignored. Thus, $(U_{t}(X_t^{\pi}))$ is a local supermartingale for all admissible portfolios and there exists an admissible portfolio for which it is a local martingale. Therefore, $U_{\cdot}(\cdot)$ satisfies conditions (b) and (c) in the definition of forward performance processes. Note that the process $(U_t(x))$ is concave and increasing in $x$, hence satisfying the condition (a) in the definition. Thus, $U_{\cdot}(\cdot)$ is a local FPP.
\end{proof}

The next result provides conditions under which the local power-mixture FPPs are indeed true FPPs.
\begin{theorem}\label{thm:true.fpp}
    Suppose the market model \eqref{eq:stock.dynamics}, a compact set of risk aversion coefficients $\mathbb{I}$, and a fixed risk aversion coefficient $\gamma_0 \in \mathbb{I}$. Let $(M^{\gamma}, V^{\gamma})$  be as in \eqref{eq:local.FPP}, \eqref{eq:local_mtg.gamma} and \eqref{eq:FV.gamma}, and let $\pi \in \A^{v}$. Additionally let $H^{\gamma_0}, J^{\gamma}$ be such that
    \begin{align}\label{eq:true.fpp.integrability.J}
        &\EE\bigg[\int_{\mathbb{I}} \exp \bigg( c_J\int_0^T |J_t^{\gamma}|^2 \dd t \bigg) \nu(\dd \gamma) \bigg]  < \infty, \quad \sup_{t\in [0,T]} \EE\bigg[\int_{\mathbb{I}}\|J_t^{\gamma}\|^{\frac{2uv}{v-1}}\nu(\dd \gamma)\bigg] < \infty,\\\label{eq:true.fpp.integrability.H}
        &\EE \bigg[ \int_{\mathbb{I}} \exp\bigg(  c_H(\gamma)\int_0^T|H_t^{\gamma_0}|^2 \dd t \bigg) \nu(\dd \gamma) \bigg] < \infty, \quad
        \sup_{t\in [0,T]} \EE[\|H_t^{\gamma_0}\|^{\frac{2uv}{v-1}}] < \infty,
    \end{align}
    for all $T>0$, where $c_J$ and $c_H(\gamma)$ are such that
    \[c_J > \frac{qp_2}{2}(qp_1 - 1) ,\quad c_H(\gamma) > \begin{cases}
    uvp_3(1-\gamma)(2uvp_1(1-\gamma) - 1)/\gamma_0^2,\\
    \frac{1}{2}qp_3\gamma\big(q p_1 \gamma - 1  \big)/\gamma_0^2,
    \end{cases}\]
    for some $p_1,p_2,p_3>1$  satisfying $\frac{1}{p_1} + \frac{1}{p_2} + \frac{1}{p_3} < 1$, and $q := \frac{2v}{v-1}$. Then, for initial preferences
    $$U_0(x) = \int_{\mathbb{I}} \frac{x^{1-\gamma}}{1-\gamma} \nu(\dd \gamma),$$
    the process
    \begin{align}\label{eq:true.FPP}
        U_t(x) &= \int_{\mathbb{I}} \frac{x^{1-\gamma}}{1-\gamma} \ee(M_t^{\gamma})\ee(V_t^{\gamma}) \nu (\dd \gamma)
    \end{align}
    is a true FPP, with an associated optimal portfolio given by a solution to
    \begin{align}\label{eq:true.FPP.portfolio}
        \sigma_t \pi_t^* = \frac{1}{\gamma_0}(\lambda_t + H_t^{\gamma_0}).
    \end{align}
\end{theorem}
The proof of this theorem is long and technical, and hence is presented in Appendix \ref{appendix.a}.
\begin{rmk}
For an individual with risk aversion $\gamma$, that is wealth preferences $U_t^{\gamma}(x)$, the full expression for the optimal utility will be
\begin{align}
    (X_t^{\pi^*})^{1-\gamma} \ee(M_t^{\gamma}) \ee(V_t^{\gamma}) =& \: \exp\bigg( \frac{1}{\gamma_0}(\lambda_s + H_s^{\gamma_0}) \cdot \dd W_s - \frac{1}{2}\int_0^t\bigg|\frac{1}{\gamma_0}(\lambda_s + H_s^{\gamma_0})\bigg|^2 \dd s \\&+ \frac{1}{\gamma_0} \int_0^t \lambda_s^T(\lambda_s + H_s^{\gamma_0}) \dd s \bigg)
\times \ee\bigg( \int_0^t \lambda_s \cdot \dd W_s \bigg) \\
&\times \ee\bigg( \int_0^t J_s^{\gamma }\cdot \dd \W_s\bigg).
\end{align}
This shows that at the optimal portfolio level all the agents will be deriving the same utility only subject to scaling according to the risk aversion coefficient $\gamma$ and market-uncorrelated utility components, that is
\begin{align}
    U_t^{\gamma}(X_t^{\pi^*}) = U_t^{\gamma_0}(X_t^{\pi^*})\times\frac{1-\gamma_0}{1-\gamma}\times\ee\bigg(\int_0^t J_s^{\gamma} \cdot \dd \W_s\bigg)\bigg/\ee\bigg(\int_0^t J_s^{\gamma_0} \cdot \dd \W_s\bigg).
\end{align}
\end{rmk}

From Theorem \ref{thm:true.fpp} we get that if the investors' preferences are shaped by all the information present in the filtration $\mathbb{F}$, then up to invertibility of $\sigma_t$ and some regularity conditions we can always choose a process $H_t$ so that any $\pi \in \A$ is deemed optimal. That is, for any initial condition $U(0, x) = \int_{\mathbb{I}}\frac{x^{1-\gamma}}{1-\gamma} \nu(\dd \gamma)$, and any portfolio $\pi \in \A$ we can explicitly construct a forward performance process, by choosing
\begin{align}\label{eq:H_portfolio}
    H_t^{\gamma_0} = \gamma_0 \sigma_t \pi_t - \lambda_t.
\end{align}
This suggests that any strategy for any person, when viewed through an appropriate lens, can be deemed dynamically consistent. Thus, the self-generation criterion (as defined in \cite{zitkovic2009}) and an initial utility datum is not enough to specify the forward development of one's preferences. We refer the reader to Section \ref{sec:power_fpp} for further discussion on this matter.

Note also that choice of the processes $\{J_t^{\gamma}\}_{\gamma \in \mathbb{I}}$ in no way affects the portfolio selection. Thus, even if we specify the optimal portfolio and the initial preferences, that will not pin down the forward performance process. This is due to assuming that all the information that affects price formation in the market can affect the investors' preferences. In other words, the filtration $\mathbb{F}$ is too large to pin down one FPP for a choice of an optimal portfolio and an initial condition. We can deal with this issue by limiting the flow of information that can affect the development of the investor's preferences. In the following subsection we do so by restricting the investor's information available to the investor to a filtration generated by some factor-driving Brownian motions $B$.
\subsection{Factor-generated Forward Performances}\label{subsec:factor_fpp}
Consider an investor whose preferences are adapted to the factor filtration $\mathbb{G}$ generated by $d_B$-dimensional Brownian motion $B_t$, such that
\begin{align}\label{eq:factor_dynamics}
    \quad B_t = \int_0^t\rho_s^{\top} \dd W_s + \int_0^t A_s^{\top} \dd \W_s,
\end{align}
where $\rho_t, A_t \in \G_t:=\sigma(B_t)$. Let the market model be as in \eqref{eq:stock.dynamics}, with $\mu_t,\sigma_t \in \G_t$. Then, for our power mixture FPPs in \eqref{eq:local.FPP}, $\{M_t^{\gamma}\}_{\gamma \in \mathbb{I}}$ will have to be adapted to $\G_t$, and hence admit the representation
\begin{align}
    M_t^{\gamma} = \int_0^t \tilde{H}^{\gamma}_s\cdot \dd B_s = \int_0^t \big(\rho_s\tilde{H}^{\gamma}_s\big)\cdot \dd W_s + \int_0^t \big(A_s\tilde{H}^{\gamma}_s\big)\cdot \dd \W_s,
\end{align}
for some $\tilde{H}_t \in \G_t$. Thus, when $U_t(x) \in \G_t$, it must be that $\tilde{H}_t^{\gamma} = (\rho_t^{\top}\rho_t)^{-1}\rho_t^{\top} H_t^{\gamma}$, and hence
\begin{align}
    J_t^{\gamma} &= A_t(\rho_t^{\top}\rho_t)^{-1}\rho_t^{\top} H_t^{\gamma}.
\end{align}
Since $$H_t^{\gamma} = \frac{\gamma - \gamma_0}{\gamma_0}\lambda_t + \frac{\gamma}{\gamma_0}H_t^{\gamma_0},$$
we get that $U_t(x)$ is completely parametrized by the process $H_t^{\gamma_0}$. Combining this with \eqref{eq:H_portfolio}, yields that to uniquely identify a factor-generated FPP it is enough to specify the initial condition and the optimal portfolio.
\begin{rmk}
Note that when the eigenvalue equality (EVE) condition (\cite[Definition 2.4]{avanesyan2020construction}) holds with a constant $p$, that is $\rho^{\top}\rho = p I_{d_B}$, we get that
\begin{align}
    \tilde{H}_t^{\gamma} = \frac{1}{r}\rho_t^{\top} H_t^{\gamma}, \quad J_t^{\gamma} = \frac{1}{r}A_t\rho_t^{\top}H_t^{\gamma}.
\end{align}
\end{rmk}
\begin{rmk}
We can certainly improve some of the lower bounds on the constants in Theorem \ref{thm:true.fpp}, given this added structure. We choose to omit these calculations as they would not be contributing anything new to the discussion.
\end{rmk}
\section{Power Forward Performances}\label{sec:power_fpp}
A particular case of our set-up is when the measure $\nu$ is a Dirac measure for some $\gamma \in (0,\infty)/\{1\}.$ Hence, we obtain the characterizations for local and true FPPs as corollaries of Theorems \ref{thm:local.FPP} and \ref{thm:true.fpp}.
\begin{corollary}\label{cor:local.FPP.power}
    Suppose the market model \eqref{eq:stock.dynamics}, and an investor with constant relative risk aversion $\gamma$. Then, for initial preferences $$U_0(x) = \frac{x^{1-\gamma}}{1-\gamma},$$
    the process
    \begin{align}\label{eq:local.FPP.power}
        U_t(x) = \frac{x^{1-\gamma}}{1-\gamma} \ee(M_t) \ee(V_t)
    \end{align}
    is a local FPP, with an associated optimal portfolio given by a solution to
    \begin{align}\label{eq:optimal.portfolio.power}
        \sigma_t\pi_t^* &= \frac{1}{\gamma}(\lambda_t + H_t),
    \end{align}
    where the pair $(M_t, V_t)$ is given by
    \begin{align}\label{eq:local_mtg.gamma.power}
         M_t &= \int_{0}^{t} H_s \cdot \dd W_s + \int_{0}^{t} J_s \cdot \dd \W_s,\quad
    V_t  = -\frac{1-\gamma}{2\gamma} \int_0^t |\lambda_s + H_s|^2  \dd s,
    \end{align}
    and $H_t, J_t \in \F_t$.
\end{corollary}
The above corollary has been stated and proved in various ways ever since \cite{musiela2008optimal}, and later being fully characterized by \cite{Choulli_2017} (including non-continuous market scenarios). Most recently this characterization was once again obtained by \cite{bo2018forward} and extended to BSDE factor representations. We note that in the latter paper when constructing true FPPs authors rely on uniform integrability assumptions, which preclude them from constructing FPPs with $H_t$ being a constant. We only require a Novikov type condition to hold, which allows us to construct such FPPs. Therefore, on the technical level the following Corollary \ref{cor:true.fpp.power} takes up its own space in the power FPP literature. 
\begin{corollary}\label{cor:true.fpp.power}
   Suppose the market model \eqref{eq:stock.dynamics} and an investor with constant relative risk aversion $\gamma$. Let $(M, V)$  be as in \eqref{eq:local_mtg.gamma.power}, and let $\pi \in \A^{v}$. Additionally let $H, J$ be such that
    \begin{align}\label{eq:true.power.fpp.integrability.J}
        &\EE\bigg[ \exp \bigg( c_J\int_0^T |J_t|^2 \dd t \bigg)  \bigg]  < \infty, &\sup_{t\in [0,T]} \EE\big[\|J_t\|^{\frac{2uv}{v-1}}\big] < \infty,\\\label{eq:true.power.fpp.integrability.H}
        &\EE \bigg[  \exp\bigg(  c_H\int_0^T|H_t|^2 \dd t \bigg)  \bigg] < \infty,
        &\sup_{t\in [0,T]} \EE\big[\|H_t\|^{\frac{2uv}{v-1}}\big] < \infty,
    \end{align}
    for all $T>0$, where $c_J$ and $c_H(\gamma)$ are such that
    \[c_J > \frac{qp_2}{2}(qp_1 - 1) ,\quad c_H > \begin{cases}
    uvp_3(1-\gamma)(2uvp_1(1-\gamma) - 1)/\gamma^2,\\
    \frac{1}{2} qp_3\gamma\big(q p_1 \gamma - 1  \big)/\gamma^2,
    \end{cases}\]
    for some $p_1,p_2,p_3>1$ satisfying $\frac{1}{p_1} + \frac{1}{p_2} + \frac{1}{p_3} < 1$, and $q := \frac{2v}{v-1}$. Then, for initial preferences 
    $$U_0(x) =  \frac{x^{1-\gamma}}{1-\gamma},$$
    the process
    \begin{align}\label{eq:true.power.FPP}
        U_t(x) &=  \frac{x^{1-\gamma}}{1-\gamma} \ee(M_t)\ee(V_t)
    \end{align}
    is a true FPP, with an associated optimal portfolio given by a solution to
    \begin{align}\label{eq:true.FPP.power.portfolio}
        \sigma_t \pi_t^* = \frac{1}{\gamma}(\lambda_t + H_t).
    \end{align}
\end{corollary}
The characterization \eqref{eq:true.power.FPP} was first implicitly derived in \cite[Theorem 4]{musiela2010portfolio2}. Without explicitly discussing the power case, the authors characterized power FPPs as a functional transformation of a time-monotone FPP obtained in \cite[Proposition 3]{musiela2010portfolio2}:
\begin{align}\label{eq:FPP_musiela}
    U_t(x) = u(x, A_t)Z_t,
\end{align}
where $A_t$ is a well-chosen finite variation process, and $Z_t = \ee(M_t)$ is the ``market-view'' process. Using the latter as a change of measure we can always consider market dynamics for which the investor's optimal decisions will be determined by a time-monotone FPP. That is, by defining a new measure $\qq$ through a Radon-Nikodym derivative
\begin{align}\label{eq:Q_measure}
    \frac{\dd \qq}{ \dd \pp}\bigg|_t = \ee(M_t), 
\end{align}
we obtain that maximizing our investor's dynamic utility $U_{\cdot}(\cdot)$ under the measure $\pp$ yields the same optimal portfolios as when maximizing the time-monotone FPP $u(x,A_t)$ under the measure $\qq$.

In addition to power forward performance processes obtained in \cite{musiela2010portfolio2}, \cite{musiela2010stochastic}, \cite{nadtochiy2014class}, \cite{nadtochiy2015optimal}, \cite{Choulli_2017}, \cite{bo2018forward}, the broad classes of FPPs derived in \cite{musiela2008optimalasset}, \cite{berrier2009characterization}, \cite{zitkovic2009}, etc. are all situated within the broad class of FPPs that have the form \eqref{eq:FPP_musiela}. Power mixture FPPs do not fall within this class. This will further become self-evident in the following section, where we fully characterize two-power mixture forward performance processes with constant power paramaters.

\subsection{Importance of the market-view process}\label{subsec.market-view}

Let us further investigate the meaning of the market-view process $\ee(M)$. In our construction, $M$ is made up of two components subjective to the investor: $H$ and $J$. Only $H$ enters the optimal portfolio selection explicitly.  As we have noted in the previous section, taking $H$ to be as in \eqref{eq:H_portfolio}, an investor can make any portfolio optimal with respect to a power FPP as in \eqref{eq:true.power.FPP}. That is, for each admissible portfolio there exists a market-view that makes it optimal. This suggests that the space of power FPPs is in fact so large that only knowing that the investor's preferences evolve as a power FPP is not enough even to narrow the search for their optimal portfolio. Any strategy, no matter how bad, has a dynamically consistent forward investment criterion justifying it. To make the forward investment problem well-posed we need to have additional information about the investor.

Thus, to reliably solve  the forward investment problem with CRRA preferences, we propose to first explicitly fix the market-view process $\ee(M)$, and solve the equivalent optimization problem of maximizing the expectation of a time-monotone FPP $x^{1-\gamma}\ee(V)/(1-\gamma)$ under the measure $\qq$ as in \eqref{eq:Q_measure}.  To justify this, let us first change the measure to a measure $\qq^H$ given through a Radon-Nikodym derivative
\begin{align}
    \frac{\dd \qq^H}{\dd \pp}\bigg|_t = \ee\bigg(\int_0^t H_s \cdot \dd W_s\bigg).
\end{align}
Then, the market dynamics will be given by
\begin{align}
    \frac{\dd S_t^i}{S_t^i} = \big\{\sigma_t^T\big((\lambda_t + H_t) \dd t + \dd W_t^H\big)\big\}_i,
\end{align}
and since the optimal portfolio is given by \eqref{eq:true.FPP.power.portfolio}, we get that under the measure $\qq^H$ the investor's optimal strategy is the traditional myopic Merton strategy. Thus, from a portfolio manager's perspective, $H$ could be interpreted as the volatility in investor's preferences due to discrepancy between the investor's and the portfolio manager's beliefs about the observable stock dynamics. Now, if we further change the measure to $\qq$ through the remaining Radon-Nikodym derivative
\begin{align}
    \frac{\dd \qq}{\dd \qq^H} \bigg|_t = \ee\bigg(\int_0^t J_s \cdot \dd \W_s\bigg)
\end{align}
the market dynamics will not be affected in an explicit way, however the distribution of $\lambda_t$ and $\sigma_t$ will change. This is best visible in a multi-factor market model setting.

Consider the eigenvalue equality (EVE) multi-factor Markovian market model as in \cite{avanesyan2020construction} with dynamics
\begin{align}
& \frac{\mathrm{d}S^i_t}{S_t^i}=\mu_i(Y_t)\,\mathrm{d}t+\sum_{j=1}^{d_W} \sigma_{ji}(Y_t)\,\mathrm{d}W^j_t,\quad i=1,\,2,\,\ldots,\,n, \\
& \dd Y_t = \alpha(Y_t) \dd t + \kappa(Y_t)^{\top}\dd B_t, \\
& B_t = \rho^{\top} W_t + A^{\top} W^{\perp}_t,\quad \rho^{\top}\rho = p \mathbf{I}_{d_B},
\end{align}
such that $d_B = d_{\W}$, and let $U_{\cdot}(\cdot) \in \G_t : = \sigma(B_t)$ be a power FPP as given in \eqref{eq:true.power.FPP}. Then, under the new measure $\qq$, the market dynamics will be
\begin{align}
& \frac{\mathrm{d}S^i_t}{S_t^i}=\Big\{\sigma(Y_t)^{\top}\Big(\big(\lambda(Y_t) + H_t\big)\,\mathrm{d}t+ \,\mathrm{d}W^{\qq}_t\Big)\Big\}_i,\quad i=1,\,2,\,\ldots,\,n, \\
& \dd Y_t = \big(\alpha(Y_t) + r^{-1}\kappa(Y_t)^{\top} \rho_t^{\top} H_t\big) \dd t + \kappa(Y_t)^{\top}\dd B_t^{\qq}.
\end{align}
Thus, the problem of forward investing is indeed reduced to maximization of the $\qq$-expectation of a time monotone performance criterion 
\begin{align}\label{eq:monotone_power_fpp}
U_t^{\qq}(x) &= \frac{1}{1-\gamma}\big(xe^{-\frac{1}{2\gamma}\|\lambda(Y_t) + H_t\|^2} \big)^{1-\gamma},
\end{align}
with market dynamics given above. From \eqref{eq:monotone_power_fpp} it follows that, under the measure $\qq$, the investor is just trying to maximize their expected power utility of wealth discounted by their perceived investment opportunities. Changing measure through the market-view process $\ee(M_t)$ gives us the investor's subjective opinion about the stock and factor dynamics (in this case the drift corrections). Thus, when discussing power FPPs that are continuous in time as well as differentiable in wealth parameters, one can always reduce the optimization problem to maximizing the expected value of a time-monotone FPP under an appropriate market-view measure.

\section{Two-Power Mixture Forward Performances}\label{sec:two-power-mixture}
Motivated by the results in \cite{Choulli_2017} we proceed to characterize a class of two-power mixture FPPs. That is we consider forward utilities of the form
\begin{align}\label{eq:fpp_mixture}
	U_t(x) = A_t x^{p_t} + D_t x^{q_t},
\end{align}
where $A, D$, and $0 < p, q < 1$ are continuous stochastic processes adapted to the filtration $\mathbb{F}$. We further assume that the power parameters remain in the order $p_t \leq q_t$ for all $t\geq0$ almost surely.

\subsection{General two-power mixture FPPs}

From the definition of forward performance processes we proceed to establish some necessary conditions the processes $A, D, p, q$ must satisfy. In particular, we obtain that $A, D$ are non-negative, and that $p, q$ are non-decreasing and non-increasing processes respectively. Hence $p, q$ must be of finite variation. 
\begin{lemma}\label{lem.non-negative-coefficients}
If a random field $U_{\cdot}(\cdot)$  given in \eqref{eq:fpp_mixture} is a forward performance process, then for all $t>0$: $A_t, D_t \geq 0$ when $\{p_t < q_t\}$, and $A_t + D_t > 0$ otherwise.
\end{lemma}
\begin{proof}
Let us fix a time $t \geq 0$. Since $U_{\cdot}(\cdot)$ is an FPP, then $U_t(x)$ must be strictly increasing in $x$. We are considering twice-differentiable functions $U_t(x)$ in $x$, hence the above statement reduces to $U_t(x)$ having a strictly positive first derivative in $x$. That is
\begin{align}
    & p_t A_t x^{p_t - 1} + q_tD_tx^{q_t - 1} > 0.
\end{align}
Since $q_t \geq p_t$, for $x>0$ we get
\begin{align}
    & p_t A_t + q_tD_tx^{q_t - p_t} > 0.
\end{align}
If $\{p_t = q_t\}$ holds, then the above inequality is equivalent to $A_t + D_t  > 0$. That is, for these realizations of the sample space, for a fixed time $t$, our investor has CRRA preferences. Now for $\{p_t < q_t\}$, then taking $x$ close to 0, and $\infty$ respectively gives us that $A_t \geq 0$, and $D_t\geq 0$. Additionally, we obtain that for no realization of the sample space  can it happen that $A_t = D_t= 0$. 
\end{proof}
\begin{rmk}
For a general discrete power mixture FPP of the form
\begin{align}
    U_t(x) = \sum_{i=1}^m A_t^{i} x^{p^i_t},
\end{align}
such that $p^1 \leq \ldots \leq p^m$ for some $m>2$, using the same approach as in the proof of Lemma \ref{lem.non-negative-coefficients}, only yields that the first and the last coefficients are non-negative $A_t^{1}, A_t^{m} \geq 0$. That is, if we consider larger discrete power mixtures than two-power mixtures, we can obtain non-negativity only for the leading risk-aversion coefficients $p^1, p^m$. This is further expanded on in Section \ref{sec:three-power}.
\end{rmk}
Now, let us show that $p$ and $q$ can only be of finite variation.
\begin{lemma}\label{lem:powers_monotonic}
Let $0 < p < q < 1$ be continuous processes and $U_t(x)$ be as in \eqref{eq:fpp_mixture}. Let $A,\, D$ be as before and such that for all $T>0$
\begin{align}\label{eq:coefficient_integrability}
    \sup_{t \in [0,T]} \E\big[ A_t \big] < \infty, \quad \sup_{t \in [0,T]} \E\big[ D_t \big] < \infty, 
\end{align}
and $U_{\cdot}(\cdot)$ is an FPP. Then, the processes $p, q$ are $\pp$-almost surely non-decreasing and non-increasing respectively. That is, for all $0<s<t$
$$p_0 \leq p_s \leq p_t < q_t \leq q_s \leq q_0, \quad \pp-a.s.$$
\end{lemma}
\begin{proof}
Take any $0<s<t$. Note that the null-portfolio, $\pi^0 := 0$, is an admissible portfolio. Thus, $(U_t(x))_{t\geq 0}$ is a supermartingale and we get 
\begin{align}\label{eq:null_portfolio}
    \E\big[ A_t x^{p_t} + D_t x^{q_t} \big| \F_s \big] \leq A_s x^{p_s} + D_s x^{q_s}.
\end{align}
First, let us consider the case when $x>1$. Then $x^{q_s} > x^{p_s}$, and since $A_s, D_s \geq 0$, \eqref{eq:null_portfolio} yields 
\begin{align}
    \E\big[ D_t x^{q_t} \big| \F_s \big] \leq \big(A_s + D_s\big)x^{q_s}.
\end{align}
Let us define an equivalent measure $\qq_1 \sim \pp$, with a Radon-Nikodym derivative
$$ \frac{\dd \qq_1}{\dd \pp} = \frac{D_t/(A_s + D_s)}{\E[D_t/(A_s + D_s)]}. $$
Thus, we obtain
$$\E^{\qq_1}\big[ x^{q_t - q_s} - 1 \big| \F_s \big] \leq C_{\qq_1} := \E\big[ D_t\big| \F_s \big]^{-1} - 1 $$
Representing $x^{q_t - q_s}$ as $e^{\log(x) (q_t - q_s)_+} + e^{\log(x) (q_t - q_s)_-} - 1 $ we get
$$ \E^{\qq_1}\big[ e^{\log(x) (q_t - q_s)_+} + e^{ - \log(x) (q_t - q_s)_-} \big| \F_s \big] \leq C_{\qq_1} + 2.$$
Since $e^{\cdot}$ is a convex function, using Jensen's inequality yields
$$ e^{\log(x) \E^{\qq_1}[(q_t - q_s)_+|\F_s]} + e^{-\log(x) \E^{\qq_1}[(q_t - q_s)_-|\F_s]} \leq C_{\qq_1} + 2.$$
Letting $x$ go to infinity we obtain that $\E[(q_t - q_s)_+ | \F_s] = 0$. Thus for all $0<s<t$ we get that $\pp-a.s.$ $q_t \leq q_s$.

\smallskip 

Now, let us consider the case $x<1$. Here, $x^{p_s}> x^{q_s}$, and just like above we get
\begin{align}
    \E\big[ A_t x^{p_t} \big| \F_s \big] \leq \big(A_s + D_s\big)x^{p_s}.
\end{align}
Defining a new measure $\qq_2 \sim \pp$
$$ \frac{\dd \qq_2}{\dd \pp} = \frac{A_t/(A_s + D_s)}{\E[A_t/(A_s + D_s)]}, $$
and proceeding as previously, we obtain that for all $x<1$
$$ e^{\log(x) \E^{\qq_2}[(p_t - p_s)_+|\F_s]} + e^{-\log(x) \E^{\qq_2}[(p_t - p_s)_-|\F_s]} \leq C_{\qq_2} + 2.$$
Letting $x$ approach $0$ we get that $\E[(p_t - p_s)_- | \F_s] = 0$, and thus for all $0<s<t$, $p_t \geq p_s,$  $\pp-$almost surely. 
\end{proof}
\begin{rmk}
By setting $x=1$ in \eqref{eq:null_portfolio}, we obtain that $(A_t + D_t)$ has to be a supermartingale.
\end{rmk}
One question that arises is, when are the finite variation processes $p$ and $q$ constant? In the proposition below we obtain that if the smaller one of the processes is constant, then the larger one has to be constant as well.
\begin{proposition}\label{prop:one_power_constant}
Let $0<p<q<1$ be continuous processes and $A_t,D_t > 0$ be continuous semimartingales such that \eqref{eq:coefficient_integrability} holds, and let
$U_{\cdot}(\cdot)$, as given in \eqref{eq:fpp_mixture}, be an FPP. Then, $q_t = q_0,$ $\pp-a.s.$ for all $t>0$ if $p_t = p_0,$ $\pp-a.s.$ for all $t>0$.
\end{proposition}
\begin{proof}
Since $A_t, D_t$ are strictly positive semimartingales, then there exist $\mathbb{F}$-adapted processes $a, a^{\perp}, d, d^{\perp}, \alpha, \delta$ such that
\begin{align}
    \dd A_t & = \alpha_t A_t \dd t + a_t A_t \cdot \dd W_t + a_t^{\perp}A_t \cdot \dd \W_t,\\
    \dd D_t & = \delta_t D_t \dd t + d_t D_t \cdot \dd W_t + d_t^{\perp}D_t \cdot \dd \W_t.
\end{align}
Let us again consider the null-portfolio $\pi^0 :=0$. As previously, we know that $U_t^{\pi^0}(x) = A_t x^{p_t} + D_t x^{q_t}$ is a supermartingale. Applying It\^{o}'s formula we get
\begin{align}
    \dd U_t^{\pi^0}(x) =&\: \log(x)\big( A_t x^{p_t} \dd p_t  +  D_t x^{q_t} \dd q_t\big) +  \big( \alpha_t A_t x^{p_t} + \delta_t D_t x^{q_t} \big) \dd t\\
    &+ \big( a_tA_tx^{p_t} + d_t D_t x^{q_t}   \big) \cdot \dd W_t + \big( a_t^{\perp} A_tx^{p_t} + d_t^{\perp} D_t x^{q_t}\big) \cdot \dd \W_t.
\end{align}

For $U_t^{\pi^0}$ to be a supermartingale it is necessary that the finite variation term is non-increasing in time. Now, let us assume that $p_t = p_0$, $\pp-a.s.$, then $\dd p_t = 0$. Thus, the above-mentioned necessary condition is equivalent to 
\begin{align}
    \log(x) \int_s^t D_r x^{q_r} \dd q_r + \int_s^t \alpha_r A_r x^{p_0} + \delta_r D_r x^{q_r} \dd r \leq 0, \quad \forall 0<s<t.
\end{align}
From Lemma \ref{lem:powers_monotonic} we get that $\dd q_r < 0$, $q_r < q_0$, and thus for $x<1$ a further necessary condition would be
\begin{align}
    \log(x) x^{q_0}\int_0^t D_r \dd q_r + \int_0^t \alpha_r A_r x^{p_0} + \delta_r D_r x^{q_r} \dd r \leq 0, \quad \forall t>0.
\end{align}
Note that since $0<q_0<1$, applying L'H\^{o}pital's rule, we get that $\log(x) x^{q_0}$ tends to $-\infty$ as $x$ goes to 0. Thus, taking $x$ to 0, we obtain that the above expression can assume positive values, unless $q_r = q_0$, $\pp-a.s.$ for all $r \in [0,t]$.
\end{proof}
This shows that if the person is confident about their relative risk aversion when they are very poor, then, in order to be a consistent investor, they have to be sure about their relative risk aversion when they are extremely rich as well. Alternatively, if the more risk averse of the two investing partners is sure of their risk aversion coefficient, then to construct a dynamically consistent investment criterion for a joint investment vehicle, the less risk averse investor must be sure of their risk aversion coefficient also.

\subsection{Constant power two-power mixture FPPs}

Having examined the case of random powers in two-power mixture FPPs we now choose to consider a constant power scenario and characterize the processes $A$ and $D$. We will focus on the case when $A$ and $D$ are strictly positive continuous semimartingales.
\begin{proposition}\label{prop:mixture_components}
Let $A, D > 0$ be continuous semimartingales adapted to $\mathbb{F}$, and let $p,q$ be constants such that $0<p<q<1$. If the process
\begin{align}\label{eq:fpp_twomixture_constant}
    U_t(x) = A_t x^p + D_t x^q 
\end{align}
is an FPP, then $A_t x^p$ and $D_t x^q$ are FPPs as well, and the optimal portfolios corresponding to all three FPPs solve the same linear system.
\end{proposition}

Please see the proof in Appendix \ref{appendix.b}. Proposition \ref{prop:mixture_components} combined with Theorems \ref{thm:local.FPP} and \ref{thm:true.fpp} provides a complete characterization of two-power mixture forward performance processes whose coefficients are strictly positive semimartingales and powers are constants in the interval (0,1).

The result can be interpreted from the lens of investment pools. Imagine we have two investors with different risk aversions, and whose preferences develop as power forward performance processes. Now, imagine that neither of them has the ability to invest in the market on the individual basis, but there is a way to invest through a joint investment vehicle. Thus, the two-power mixture FPPs are characterizing the joint utility the investors derive from their invested capital's performance. Proposition \ref{prop:mixture_components} shows that the investors will be just as happy investing together as if they had the opportunity to invest apart only if their choices would have been the same anyways. From the proof we can see that the drift of $U_{\cdot}(\cdot)$, as defined in \eqref{eq:fpp_twomixture_constant}, will always be non-positive as long as the components are FPPs themselves
\begin{align}
    -\frac{pq(1-p)(1-q)A_tD_t(X_t^{\pi^*})^{p+q}}{p(1-p)A_t (X_t^{\pi^*})^p + q(1-q)D_t (X_t^{\pi^*})^q}\bigg|\frac{1}{1-p}(\lambda_t + a_t) - \frac{1}{1-q}(\lambda_t + d_t) \bigg|^2
     \leq 0.
\end{align}
That is, there is no cost to pooling resources together only if the investors have identical individual optimal strategies. 

\subsection{Non-FPP forward strategies}\label{subsec:non-fpp}
One of the defining properties that makes forward performance processes useful is that they guarantee existence of a dominating strategy by making the utility random field a martingale at such a strategy and a supermartingale for all other admissible strategies. As we have observed in Proposition \ref{prop:mixture_components}, when dealing with pooled resources, in most of the cases the joint utility random field cannot be a martingale. Does this mean that there is no dominating strategy? That is, if two investors' ideal strategies do not perfectly align, should they always pass up on the investment opportunity? Or maybe there is a strategy that will not keep the joint utility at a prior level, however could be close to optimal. We believe that in some cases $\pi^*$ given by \eqref{eq:mixture_portfolio} could be a good approximation.

Consider two investors with the same market-view, and whose preferences develop as time-monotone power FPPs with power coefficients $p<q$. That is
\begin{align}
    U^1_t(x) &= A_0 x^p e^{-\frac{p}{2(1-p)}\int_0^t |\lambda_s|^2\dd s},\quad U^2_t(x) = D_0 x^q e^{-\frac{q}{2(1-q)}\int_0^t |\lambda_s|^2\dd s}.
\end{align} 
When faced with the market dynamics as in \eqref{eq:stock.dynamics}, we know that each respective investor's optimal allocation within the risky investment vehicles would be identical. They would only differ in the percentage of their wealth that their respective optimal portfolio allocations within these risky assets, $\pi^1, \pi^2$ would be satisfying
\begin{align}
    \sigma_t\pi_t^{1,*} = \frac{1}{1-p}\lambda_t, \quad     \sigma_t\pi_t^{2,*} = \frac{1}{1-q}\lambda_t.
\end{align}
Thus, to appease both of the investors, we will be looking for strategies that have identical risky asset allocation to the individual investors' strategies, and should only differ in the coefficient. That is, our admissible portfolio space will be restricted to
\begin{align}
    \A_{\lambda}^{v} = \bigg\{\pi \in \A^{v} \bigg| \sigma_t\pi_t = \frac{1}{1 - z_t} \lambda_t, \textrm{ for some } z_t \in (0,1) \textrm{ adapted to } \F_t\bigg\}.
\end{align}
For a portfolio 
\begin{align}
    \sigma_t\pi_t^z = \frac{\lambda_t}{1-z_t}
\end{align}
the joint utility will take the shape 
\begin{align}
    U_t\big(X_t^{\pi^z}\big) =&\: A_0 X_0^p \ee\bigg(\frac{p}{1-z}\int_0^t\lambda_s \cdot \dd W_s\bigg) \exp\bigg(-\frac{p}{2(1-p)}\int_0^t \frac{(z_s-p)^2}{(1-z_s)^2} |\lambda_s|^2 \dd s\bigg) \\
    &+ D_0 X_0^q \ee\bigg(\frac{q}{1-z}\int_0^t\lambda_s \cdot \dd W_s\bigg) \exp\bigg(-\frac{q}{2(1-q)}\int_0^t \frac{(z_s-q)^2}{(1-z_s)^2} |\lambda_s|^2 \dd s\bigg).
\end{align}
For exposition purposes we limit ourselves to a one-stock complete market case scenario, where the stock price develops as a geometric Brownian motion
\begin{align}
    \frac{\dd S_t}{S_t} = \sigma_t\big( \lambda \dd t + \dd W_t \big).
\end{align}
Let us consider a subclass of $\A_{\lambda}$, where the coefficient is constant for all time $t>0$, $z_t = z \in (0,1)$. That is, the investors come to an agreement about a constant portion of the joint wealth to invest in the risky asset, and keep that proportion until the end of time. In that case, the expression for the joint utility will be
\begin{align}
    U_t\big(X_t^{\pi^z}\big) =&\: A_0 X_0^p \ee\bigg(\frac{p}{1-z}\lambda \cdot W_t\bigg) \exp\bigg(- \frac{p(z-p)^2\lambda^2 t}{2(1-p)(1-z)^2}\bigg) \\
    &+ D_0 X_0^q \ee\bigg(\frac{q}{1-z} \lambda \cdot W_t\bigg) \exp\bigg(- \frac{q(z-q)^2\lambda^2 t}{2(1-q)(1-z)^2}\bigg).
\end{align}
From here, we can calculate the explicit expression for the expectation of the forward utility
\begin{align}\label{eq:constant_expectation} \begin{split}
    \E\big[U_t(X_t^{\pi^z})\big]  =& \: A_0 X_0^p \exp\bigg( - \frac{p(z-p)^2\lambda^2 t}{2(1-p)(1-z)^2} \!\bigg) \! \\
    &+ \! D_0 X_0^q \exp\bigg(\! - \frac{q(z-q)^2\lambda^2 t}{2(1\!-\!q)(1\!-\!z)^2} \! \bigg). 
    \end{split}
\end{align}
The expression above is not a concave function of $z$ and does not even always have a unique optimizer. In fact, keeping $p, A_0$ and $q, D_0$ fixed, the optimizer in general will change as a function of $X_0$, $t$, and $\lambda$.

\subsubsection{Numerical illustrations}

Let us illustrate some examples of the functional shape in \eqref{eq:constant_expectation}. Below we plot all constant-proportion strategies 30 years into the future. Let $A_0 = D_0 = 1$, we will change the values of $p,q,\lambda,X_0$ and observe how that affects the shape of the expectation.
\begin{itemize}
    \item $p = 0.1, q = 0.3, \lambda = 1, X_0 = 1$.
    \begin{figure}[ht!]
    \includegraphics[width = \textwidth]{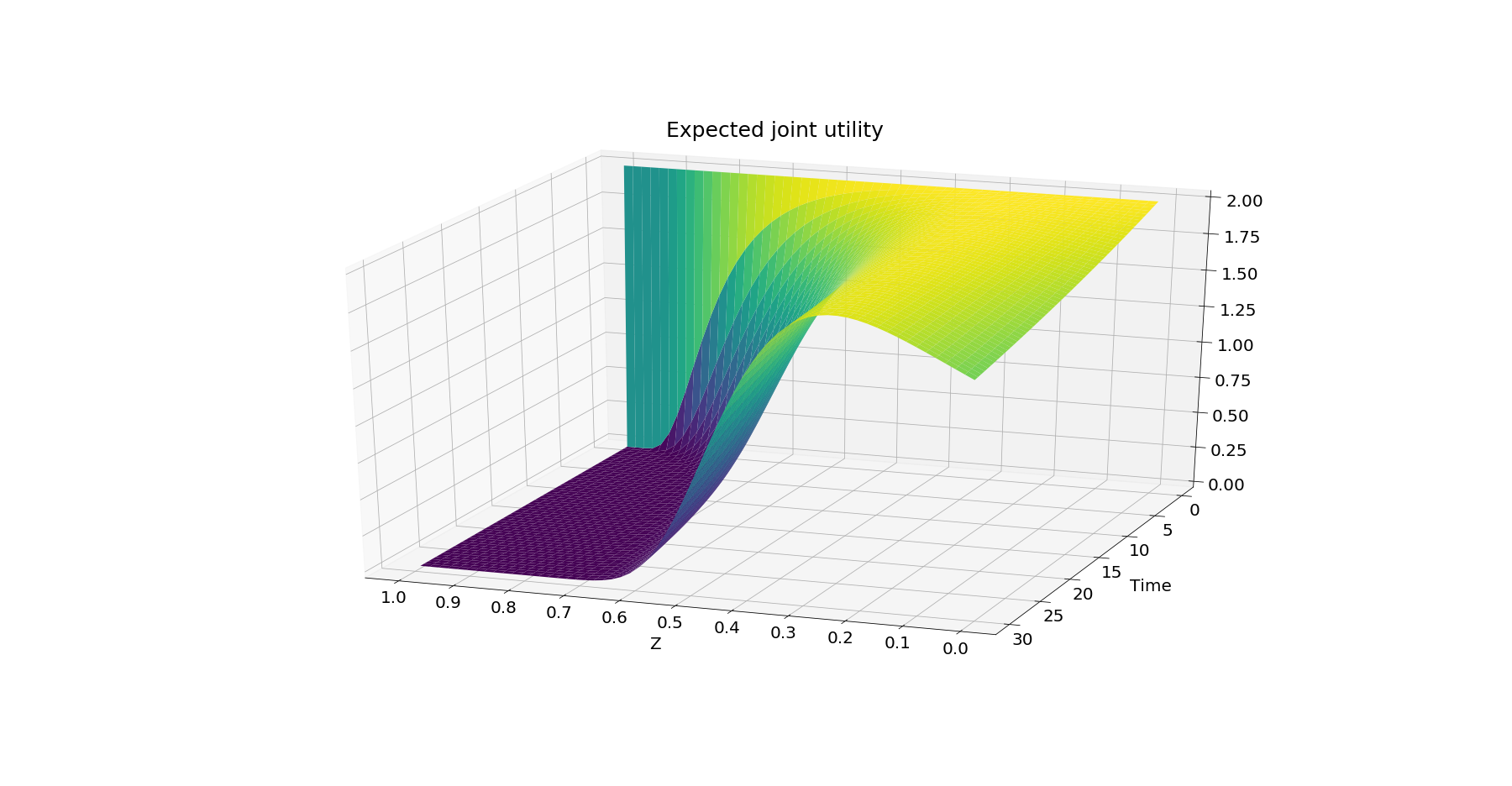}
    \caption{Expected joint utility throughout time for parameter values $p = 0.1, q = 0.3, \lambda = 1, X_0 = 1$}\label{fig:exp_util}
    \end{figure}
    As we can see in Figure~\ref{fig:exp_util}, for the 30 year period there seems to be a quite stable optimizer located at at around $z = 0.25$ for all times $t \in [0,30]$. Increasing $X_0$ pushes this optimizer further toward the value of $q$, with the wealth $X_0 > 1000$ resulting in $z^* \approx 0.3$. Unsurprisingly, making $X_0$ go to $0$ results in $z^*$ approaching the value of $p$. The initial wealth has to be very small $X_0 < 0.001$ for $z^*$ to get close enough to 0.1. In any case one general trend that we observe is that the optimal decision always results in compromise. That is, it is jointly beneficial for both parties to abandon their individual optimal allocations and meet somewhere in the middle. 
    \item $p = 0.1, q = 0.3, \lambda = 0.5, X_0 = 1$. 
    \begin{figure}[ht!]
    \includegraphics[width = \textwidth]{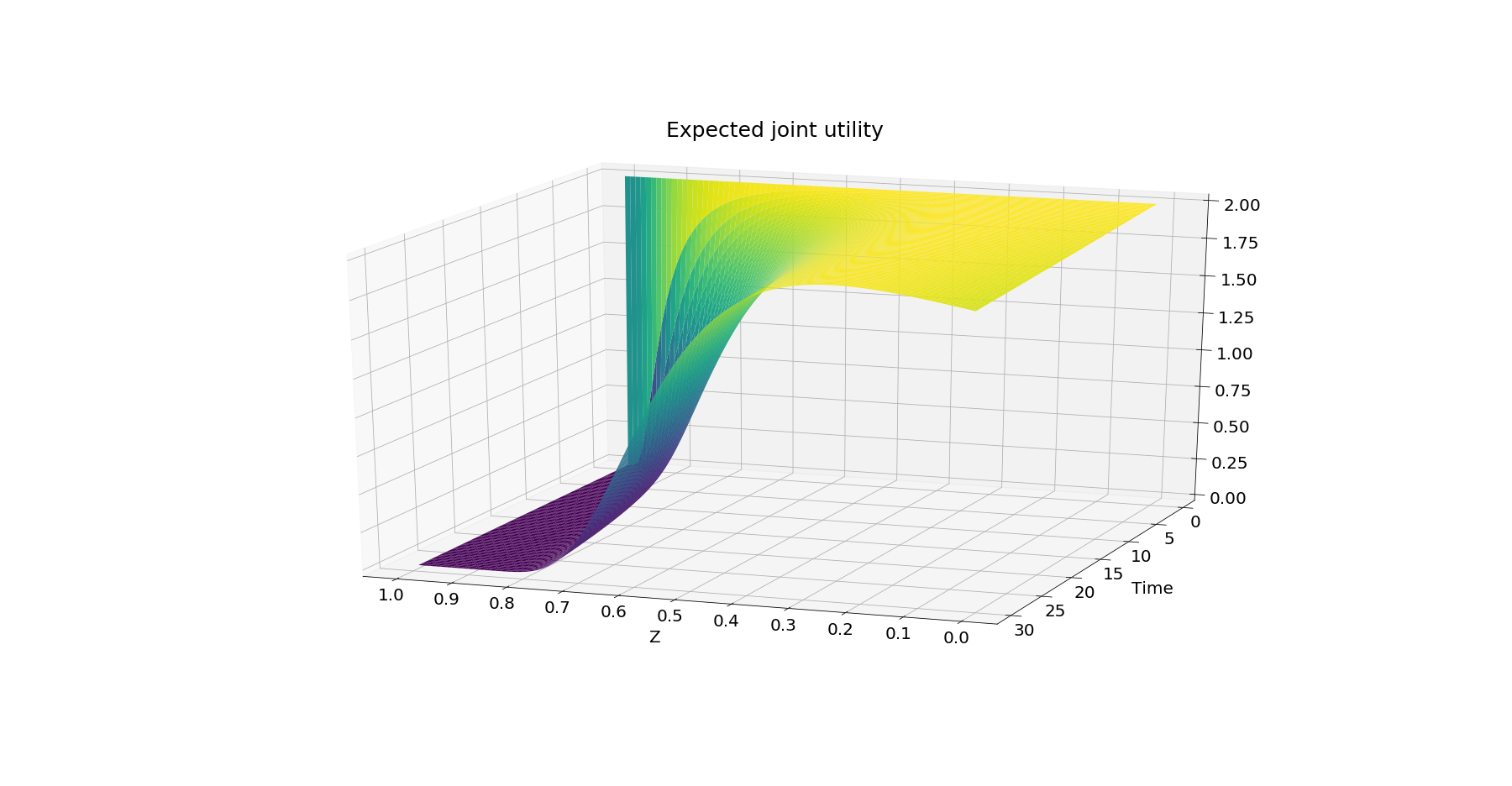}
    \caption{Expected joint utility throughout time for parameter values $p = 0.1, q = 0.3, \lambda = 0.5, X_0 = 1$}\label{fig:exp_util_sharpe_low}
    \end{figure}
    Making the investment opportunity less lucrative by setting Sharpe to be 0.5, shifts the compromise exactly towards the center $z^* \approx 0.2$ (see 
    Figure~\ref{fig:exp_util_sharpe_low}). That is, a more risk averse investor's opinion matters more when there is less promise of rags-to-riches.
    \item $p = 0.1, q = 0.3, \lambda = 4, X_0 = 1$. 
    \begin{figure}[ht!]
    \includegraphics[width = \textwidth]{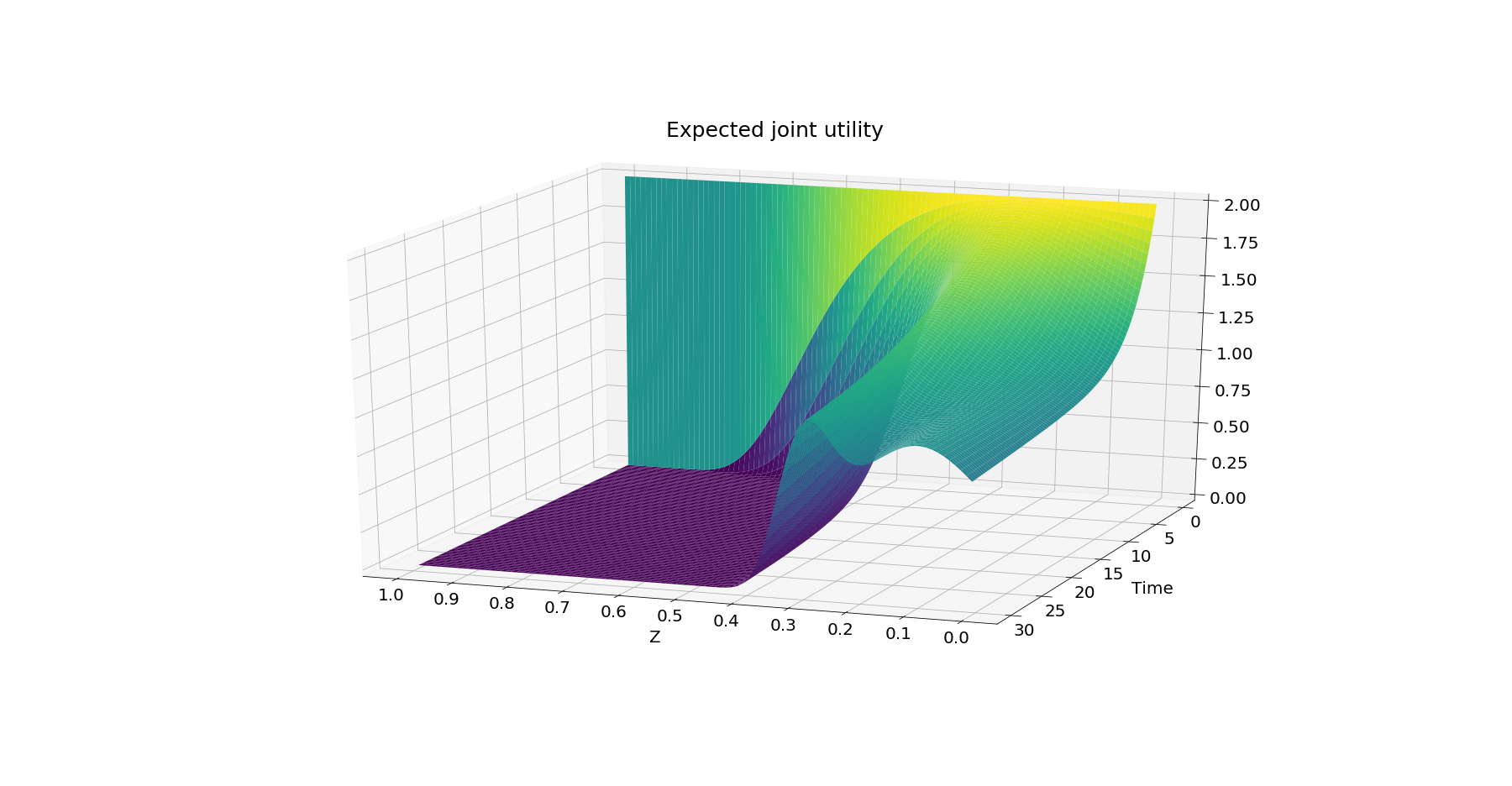}
    \caption{Expected joint utility throughout time for parameter values $p = 0.1, q = 0.3, \lambda = 4, X_0 = 1$}\label{fig:exp_util_sharpe_high}
    \end{figure}
    Changing the Sharpe ratio significantly has changed the shape of the curve in Figure~\ref{fig:exp_util_sharpe_high} so that there are two local maxima now, one located close to $p$, and the other close to $q$ respectively. The optimal allocation is still one requiring compromise, however it is further shifted towards the optimal allocation of the less risk averse investor. In particular, what we see is that in the long run the more risk averse investor would have to bend to the will of their more risk loving peer, which might result in the arrangement crumbling. That is, when the investment opportunity is particularly lucrative, people would become less willing to compromise over the long term. Finally, the shape is also highly dependent on the initial wealth parameter. Setting $X_0 = 0.1$ will result in the bump at $z = 0.1$ becoming the global optimum, whereas setting $X_0 = 1000$ would almost completely smooth this bump out. Either way we do observe that the expected value of the joint utility decreases much more than in the previous two cases. That is, the price of cooperation is particularly high when the investment is lucrative.
     \item $p = 0.1, q = 0.6, \lambda = 1, X_0 = 1$. 
     \begin{figure}[ht!]
    \includegraphics[width = \textwidth]{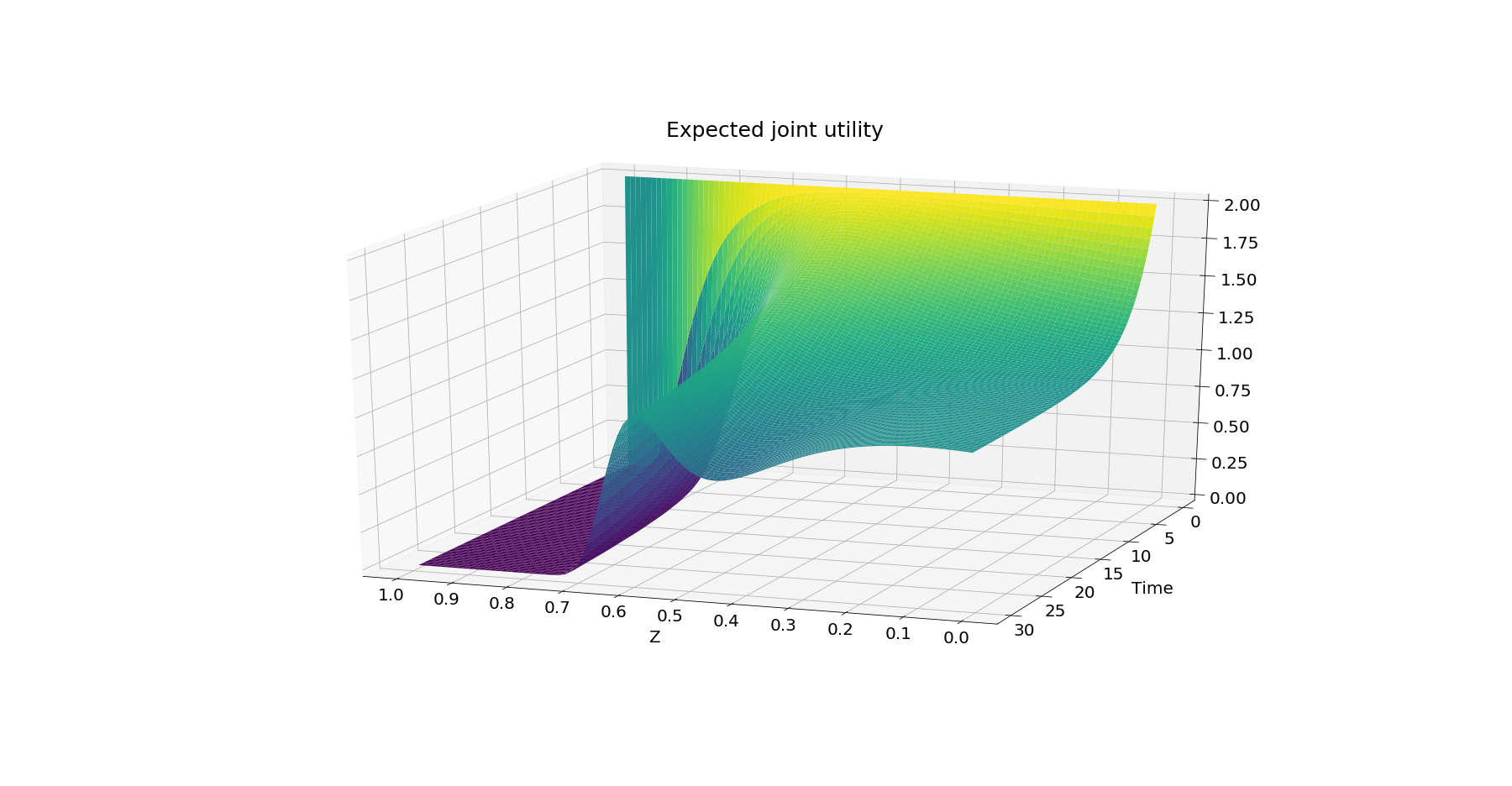}
    \caption{Expected joint utility throughout time for parameter values $p = 0.1, q = 0.6, \lambda = 1, X_0 = 1$}\label{fig:exp_util_risk_diff}
    \end{figure}
    Finally, letting the risk aversion gap increase deems long-term pooling of investment unreasonable as can be seen in Figure~\ref{fig:exp_util_risk_diff}. This shows that when pooling investment, choosing people who do not differ too much in terms of their risk aversions would result in longer willingness to keep investing together. That is, for the sake of keeping the collective togetherness longer, carefully selecting the partner prior to starting any venture could play a key role.
 \end{itemize}

 Having analyzed the constant-portion investing from different angles we come to a conclusion that given an investment opportunity, one can find investment partners with an appropriate risk aversion differential so that long-term cooperation is optimal.

 From our prior analysis it looks like cooperation does pay off over a 30-year period when $p = 0.1, q = 0.3, \lambda = 1, X_0 = 1$. Constant portion investing results in $z^* = 0.25$ for these parameter values. Now, we would like to compare constant portion investing to two investment strategies: $\pi^*$ from \eqref{eq:mixture_portfolio}, and $\pi^e$ which will be one-period expected utility maximizing portfolio (one period conditional version of \eqref{eq:constant_expectation}). We set the investment horizon to 30 years (periods), and rebalance the portfolios every year. We perform the simulation 1,000 times and present in Figure~\ref{fig:utility_comparison} the average paths of the respective joint utilities.
 \begin{figure}[h]
    \includegraphics[width = \textwidth]{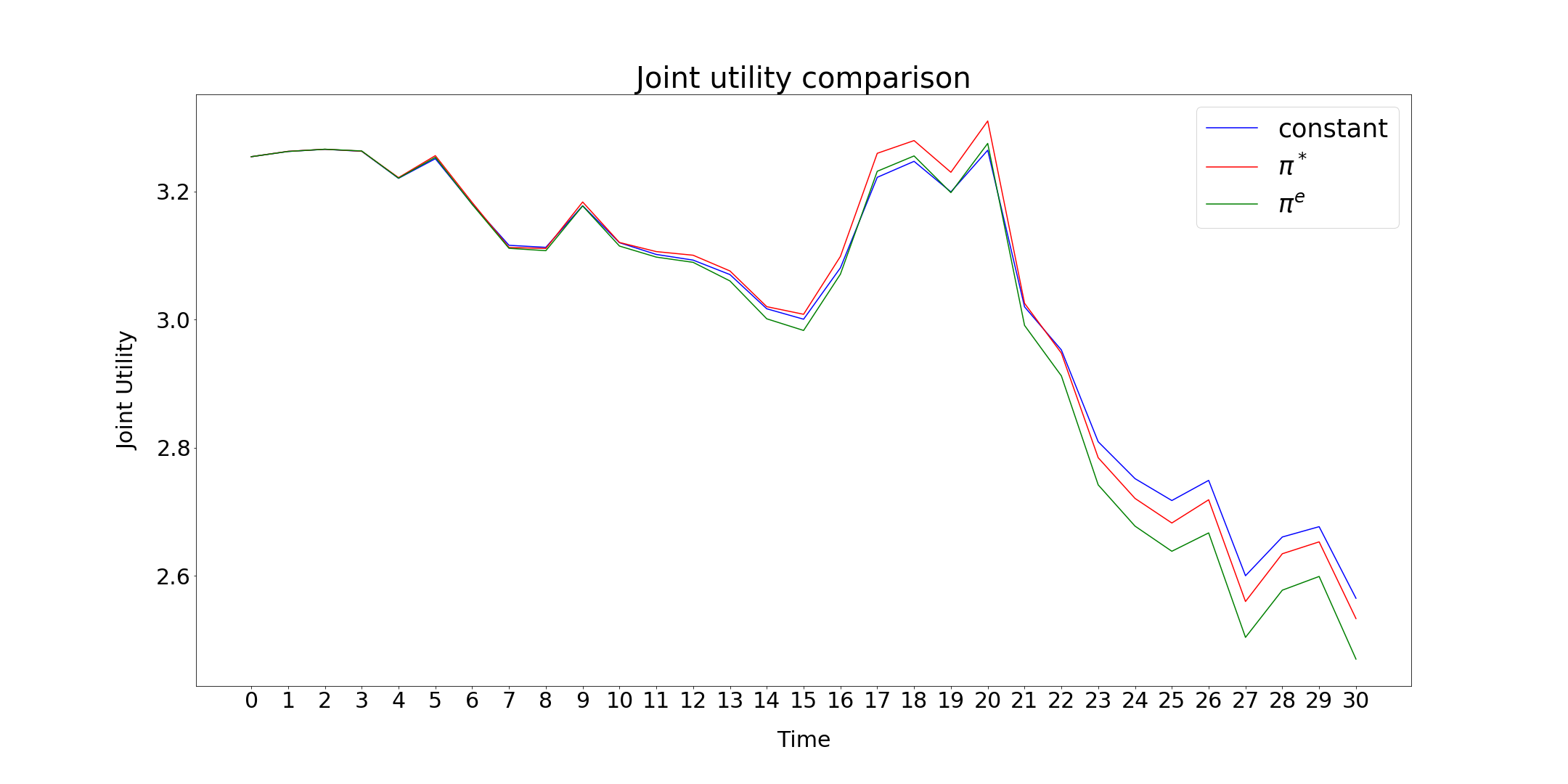}
    \caption{Forward utility performance of constant, $\pi^*$ and $\pi^e$ strategies}\label{fig:utility_comparison}
    \end{figure}
 \begin{figure}[h!]
    \includegraphics[width = \textwidth]{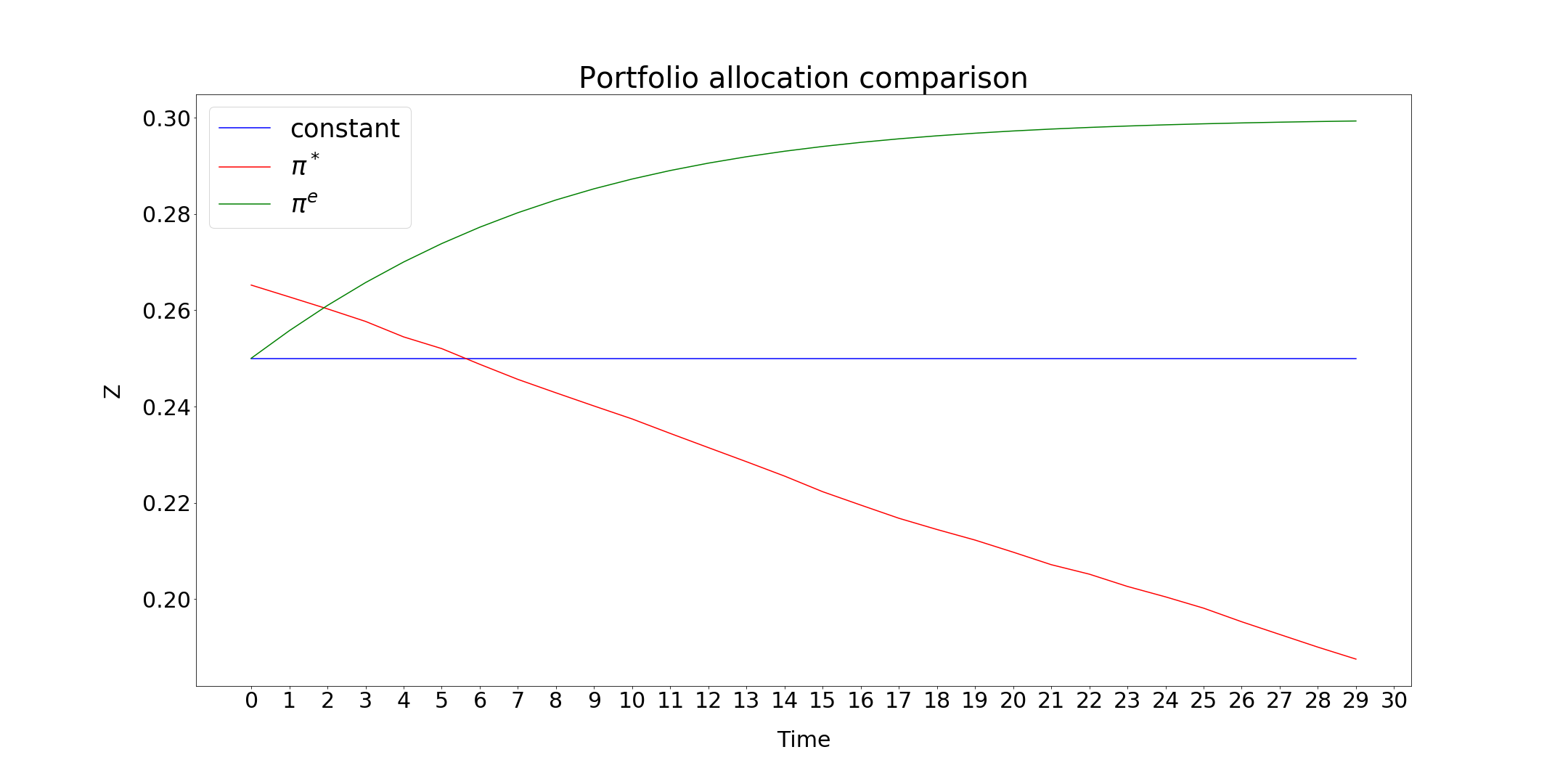}
    \caption{Portfolio allocations of the constant, $\pi^*$ and $\pi^e$ strategies throughout time.}\label{fig:portfolio_comparison}
    \end{figure}
As we can see the well-chosen constant allocation mostly dominates the other two strategies. It however places too much emphasis on satisfying the more risk loving investing partner. The strategy $\pi^e$ is even more guilty of such preferential treatment. As we can see in Figure~\ref{fig:portfolio_comparison}, the most fair choice of portfolio allocations into risky assets occurs when exercising the strategy $\pi^*$. As can be seen in Figure~\ref{fig:utility_comparison}, $\pi^*$ does not perform much worse than the constant optimal portion, and hence could be the investment strategy of choice.

\subsection{Explicit two power dual example}
In \cite{zitkovic2009} and \cite{berrier2009characterization} maturity-independent risk measures were constructed as convex duals of corresponding FPPs. We wrap up this section by discussing an example where we can find the dual of a two-power mixture FPP explicitly, thereby constructing new maturity-independent risk measures.

Consider a two-power mixture FPP with component risk aversions $\gamma$ and $2\gamma$
\begin{align}\label{eq:two-power_fpp_double_aversion}
    U_t(x) = A_t\frac{x^{1-2\gamma}}{1-2\gamma} + D_t\frac{x^{1-\gamma}}{1-\gamma}.
\end{align}
From Proposition \ref{prop:mixture_components} we get that $A_t,D_t$ must satisfy \eqref{eq:two-power_A_dynamics}, \eqref{eq:two-power_D_dynamics}, and \eqref{eq:two-power_fpp_FV}. Now let's consider the Legendre dual of $U_t(x)$, for $y\geq 0$
\begin{align}
    V_t(y) = \sup_{x \geq 0} A_t\frac{x^{1-2\gamma}}{1-2\gamma} + D_t\frac{x^{1-\gamma}}{1-\gamma} - xy.
\end{align}
Let us find the explicit expression for $V_t(y)$. Since $U(t,x)$ is globally concave and differentiable in $x$, we can find the supremum by equalizing to zero the first derivative
\begin{align}
    x_*^{-2\gamma}A_t + x_*^{-\gamma}D_t - y = 0.
\end{align}
Solving the quadratic equation for $x_*^{-\gamma}$ we finally get that
\begin{align}
    x_* = \bigg(\frac{-D_t + \sqrt{D_t^2 + 4 A_t y}}{2A_t}\bigg)^{-\frac{1}{\gamma}}.
\end{align}
Hence, by plugging in the expression for $x^*$ we can obtain the explicit expression for the dual $V_t(\cdot)$.
\section{Three-Power Mixture Forward Performance Construction}\label{sec:three-power}
The main results of the previous section are concerning necessary conditions two-power FPPs must satisfy. Lemma \ref{lem.non-negative-coefficients} and Proposition \ref{prop:mixture_components} show that we can only form two-power mixture FPPs as positive linear combinations of power FPPs. This, however, is not a necessary condition for general power mixture FPPs. We illustrate this point below by constructing a three-power mixture FPP from a non-positive linear combination of power FPPs.

Consider the process 
\begin{align}\label{eq:three-power-mixture}
\begin{split}
    U_t(x) =& \: \frac{x^{1-\gamma}}{1-\gamma} \ee\big(M_t^{\gamma}\big)\ee\big(V_t^{\gamma}\big) - \frac{x^{1-2\gamma}}{1-2\gamma} \ee\big(M_t^{2\gamma}\big)\ee\big(V_t^{2\gamma}\big)\\
    &+ 
    \frac{x^{1-3\gamma}}{1-3\gamma} \ee\big(M_t^{3\gamma}\big)\ee\big(V_t^{3\gamma}\big),
\end{split}
\end{align}
where $\gamma \in (0, 1/3)$, and $M_t^{\cdot}, V_t^{\cdot}$ are given as in equations \eqref{eq:local_mtg.gamma}, \eqref{eq:FV.gamma}, with 
\[\gamma_0 = 2\gamma, \quad H_t^{\gamma_0} = H_t^{2\gamma} = 0,\quad J_t^{\cdot} = 0.\]
Plugging these into \eqref{eq:three-power-mixture} we get 
\begin{align}
    U_t(x) =& \: \frac{x^{1-\gamma}}{1-\gamma} \ee\bigg(- \frac{1}{2} \int_0^t \lambda_s \cdot \dd W_s\bigg) e^{-\frac{1-\gamma}{8\gamma}\int_0^t \|\lambda_s\|^2 \dd s} - \frac{x^{1-2\gamma}}{1-2\gamma} e^{-\frac{1-2\gamma}{4\gamma}\int_0^t \|\lambda_s\|^2 \dd s} \\
    &+ \frac{x^{1-3\gamma}}{1-3\gamma} \ee\bigg( \frac{1}{2} \int_0^t \lambda_s \cdot \dd W_s\bigg)_t e^{-\frac{3-9\gamma}{8\gamma}\int_0^t \|\lambda_s\|^2 \dd s}.
\end{align}
Denoting $e^{\frac{1}{2}\int_0^t \lambda_s^T \dd W_s}:= Z_t$ yields
\begin{align}
    U_t(x) =& \: Z_t^{-1} \bigg( \frac{x^{1-\gamma}}{1-\gamma} e^{-\frac{1}{8\gamma}\int_0^t \|\lambda_s\|^2 \dd s} - \frac{x^{1-2\gamma}}{1-2\gamma} e^{-\frac{1-2\gamma}{4\gamma}\int_0^t \|\lambda_s\|^2 \dd s} Z_t \\
    &+ \frac{x^{1-3\gamma}}{1-3\gamma} e^{(1 - \frac{3}{8\gamma})\int_0^t \|\lambda_s\|^2 \dd s} Z_t^2\bigg).
\end{align}
Now, to show that $U_{\cdot}(\cdot)$ is indeed an FPP let us first show that it is strictly increasing and concave in $x$. Taking the first and second derivatives in $x$ yields
\begin{align}
    \partial_x U_t(x) =& \: Z_t^{-1}x^{-\gamma}\big( e^{-\frac{1}{8\gamma}\int_0^t \|\lambda_s\|^2 \dd s} -  e^{-\frac{1-2\gamma}{4\gamma}\int_0^t \|\lambda_s\|^2 \dd s} x^{-\gamma}Z_t + e^{(1 - \frac{3}{8\gamma})\int_0^t \|\lambda_s\|^2 \dd s} x^{-2\gamma}Z_t^2 \big),\\
    \partial_x^2 U_t(x) =& \: -  \gamma Z_t^{-1}x^{-\gamma - 1}\big( e^{-\frac{1}{8\gamma}\int_0^t \|\lambda_s\|^2 \dd s} -  2 e^{-\frac{1-2\gamma}{4\gamma}\int_0^t \|\lambda_s\|^2 \dd s} x^{-\gamma}Z_t\\ 
    &+ 3 e^{(1 - \frac{3}{8\gamma})\int_0^t \|\lambda_s\|^2 \dd s} x^{-2\gamma}Z_t^2 \big).
\end{align}
Thus, we get that $U_{t}(x)$ is strictly increasing and concave in $x$ if and only if 
\begin{align}
    & e^{-\frac{1}{8\gamma}\int_0^t \|\lambda_s\|^2 \dd s} -  e^{-\frac{1-2\gamma}{4\gamma}\int_0^t \|\lambda_s\|^2 \dd s} x^{-\gamma}Z_t + e^{(1 - \frac{3}{8\gamma})\int_0^t \|\lambda_s\|^2 \dd s} x^{-2\gamma}Z_t^2 > 0,\\
    &e^{-\frac{1}{8\gamma}\int_0^t \|\lambda_s\|^2 \dd s} -  2 e^{-\frac{1-2\gamma}{4\gamma}\int_0^t \|\lambda_s\|^2 \dd s} x^{-\gamma}Z_t + 3 e^{(1 - \frac{3}{8\gamma})\int_0^t \|\lambda_s\|^2 \dd s} x^{-2\gamma}Z_t^2 > 0 .
\end{align}
The above expressions are both quadratic polynomials in $x^{-\gamma}Z_t$, with respective discriminants $- 3 e^{- (1-2\gamma)/2\gamma}$ and $- 8 e^{- (1-2\gamma)/2\gamma}$. Since both the discriminants are negative, we obtain the desired inequalities. Thus, $U_{\cdot}(\cdot)$ is strictly increasing and concave in the wealth parameter.

Now, let us show that $(U_t(X_t^{\pi}))$ is a supermartingale for all admissible $\pi \in \mathcal{A}^v$, and that it is a martingale for some $\pi^* \in \mathcal{A}^v$. For sake of convenience denote
\begin{align}
    A_t & :=  \ee\big( M_t^{\gamma} \big)\ee\big( V_t^{\gamma} \big), \quad C_t :=  \ee\big( M_t^{2\gamma} \big)\ee\big( V_t^{2\gamma} \big), \quad
    D_t := \ee\big( M_t^{3\gamma} \big)\ee\big( V_t^{3\gamma} \big).
\end{align}
Applying It\^{o}'s lemma yields
\begin{align}
    \dd U_t\big(X_t^{\pi}\big) =& \: - Z_t^{-1}\big(X_t^{\pi})^{1-\gamma}\Big(  e^{-\frac{1}{8\gamma}\int_0^t \|\lambda_s\|^2 \dd s} - 2 e^{-\frac{1-2\gamma}{4\gamma}\int_0^t \|\lambda_s\|^2 \dd s} Z_t \big(X_t^{\pi}\big)^{-\gamma} \\
    &+ 3e^{(1 - \frac{3}{8\gamma})\int_0^t \|\lambda_s\|^2 \dd s} Z_t^2 \big(X_t^{\pi}\big)^{-2\gamma}  \Big)
    \times \bigg( \frac{1}{8 \gamma}\|\lambda_t\|^2  - \frac{1}{2}(\sigma_t\pi_t)^{\top}\lambda_t + \frac{\gamma}{2} \|\sigma_t \pi_t\|^2 \bigg) \dd t \\
    & + \big(X_t^{\pi}\big)^{1-\gamma}\Big( A_t - C_t \big(X_t^{\pi}\big)^{-\gamma} + D_t \big(X_t^{\pi}\big)^{-2\gamma}\Big)(\sigma_t\pi_t)^{\top} \dd W_t \\
    &- \frac{1}{2(1-\gamma)}\big(X_t^{\pi}\big)^{1-\gamma} A_t \lambda_t^{\top} \dd W_t  +  \frac{1}{2(1-3\gamma)}\big(X_t^{\pi}\big)^{1-3\gamma} D_t \lambda_t^{\top} \dd W_t.
\end{align}
Let us first consider the drift term. Note that it is written as a multiplication of two quadratic polynomials in $Z_t(X_t^{\pi})^{-\gamma}$ and $\sigma_t\pi_t$ respectively. The former polynomial has a negative discriminant and thus is strictly positive. The latter polynomial is strictly convex and admits the minimizer
\begin{align}
\sigma_t \pi_t^* &= \frac{1}{2\gamma} \lambda_t.
\end{align}
Plugging it back into the drift term produces the value of 0. Thus, for all admissible $\pi$ the drift is negative, and is equal to 0 for the portfolio $\pi^*$. Therefore $U_t(X_t^{\pi})$ is a local supermartingale and $U_t(X_t^{\pi^*})$ is a local martingale.

Finally, as $\lambda$ is bounded and $\pi \in \mathcal{A}^v$, following the same steps as in the proof of Theorem \ref{thm:true.fpp} we obtain that the local martingale term is indeed a martingale, and that $\pi^*$ is indeed admissible. Thus, $U_{\cdot}(\cdot)$ is a true FPP, and therefore we have constructed a three-power mixture that is not a positive linear combination of power FPPs. That is, the necessary conditions that we have obtained for the two-power mixtures do not apply to the more general power-mixtures.

Moreover, this example covers a class of admissible initial conditions 
\[ U_0(x) = \frac{x^{1-\gamma}}{1-\gamma} - \frac{x^{1-2\gamma}}{1-2\gamma} + \frac{x^{1-3\gamma}}{1-3\gamma} \quad s.t.\: \gamma \in \bigg(0,\frac{1}{3}\bigg)\]
for constructing true FPPs, that is not covered by Theorems \ref{thm:local.FPP} and \ref{thm:true.fpp}. That is, the three-power mixture example shows that to construct a power mixture FPP, the measure $\nu$ does not have to be positive. This stands in contrast to the two-power mixture case where the measure $\nu$ has to be positive (see Lemma \ref{lem.non-negative-coefficients}). 
\begin{rmk}
All the summands in this three-power mixture 
$$\frac{A_t}{1-\gamma} x^{1-\gamma},\quad  \frac{C_t}{1-2\gamma} x^{1-2\gamma},\quad \frac{D_t}{1-3\gamma} x^{1-3\gamma}$$
are true FPPs themselves. The coefficients of the summands do not change sign in this construction of a three-power mixture FPP. It might be possible to obtain this observation as a necessary condition for the general power mixture case. We leave this for future investigation.
\end{rmk}

\section{Conclusion}\label{sec:conclusion}
In this work we have considered the problem of forward optimal investment in a continuous semimartingale market model adapted to a Brownian filtration. We have introduced and constructed the broad class of power mixture forward performance processes in Theorems \ref{thm:local.FPP} and \ref{thm:true.fpp}. As a direct result of these theorems we established a new class of power true FPPs in Corollaries \ref{cor:local.FPP.power} and \ref{cor:true.fpp.power}, and provided a discussion of the parametrization. For two-power mixtures we have obtained in Lemma \ref{lem:powers_monotonic} that the powers can only be finite-variation processes, and if the smaller power is constant in Proposition \ref{prop:one_power_constant}, we have established that the other power should be constant too. In Proposition \ref{prop:mixture_components}, we completely characterized the two-power mixture FPPs with constant power coefficients. Finally, in Section \ref{sec:three-power} we constructed a three-power mixture FPP that exhibits the limitations of extending our obtained necessary conditions for the two-power mixtures to the general case. Future directions could include obtaining more relaxed necessary conditions for the general power mixture case.

Power mixture FPPs can be interpreted as the joint consistent utility derived by a pool of investors with different risk aversions. A particular case of interest is when dealing with pooling investment decisions of two investors whose initial utilities are of power form. In Proposition \ref{prop:mixture_components} we have obtained that to have their joint dynamic utility develop as an FPP, the individual investors' preferences must evolve as power FPPs themselves and must admit identical optimal strategies. So, in the two-investor case, there is no cost to cooperation only when the investors would have made the same decisions anyways. This led us to question what happens when the two investors have different optimal portfolio allocations. In Section \ref{subsec:non-fpp} we discussed the case of two investors who share the same ``market-view'' but have different risk aversions. We observed that if the investor's risk aversions are close enough, then setting up a joint investment pool could be a reasonable proposition. We leave the analysis of the joint investment problem when the investors do not share the same ``market-view'' for future work.

What makes forward performance processes useful is the fact that we can obtain the optimal investment strategy through the martingale-supermartingale structure of the utility random field. A curious question for future investigation is whether such structure is necessary to have a dominating investment strategy. That is, an interesting line of future work could be trying to construct  a large class of random fields that are strict supermartingales for all admissible portfolios, but still admit a provably optimal strategy.

\appendix
\section{Proof of Theorem \ref{thm:true.fpp}}\label{appendix.a}
Before we proceed to prove Theorem \ref{thm:true.fpp}, first let us prove a useful result.
\begin{lemma}\label{lem:exp.integrability}
Consider a compact set $\mathbb{I} \subset (0, \infty)/\{1\}$, a positive finite measure $\nu$, and  two families of stochastic processes $\{X_t^{\gamma}\}_{\gamma \in \mathbb{I}},\, \{b_t^{\gamma}\}_{\gamma \in \mathbb{I}}\in \F_t $ such that for some $T, c>0$ 
$$\sup_{t\in[0,T], \gamma \in \mathbb{I}} |b_t^{\gamma}|^2 < \infty, \quad \EE\bigg[\int_{\mathbb{I}} \exp\bigg(c\int_{0}^T|X_t^{\gamma}|^2\dd t\bigg)\nu(\dd \gamma) \bigg] < \infty.$$
Then, for all $c_1 \in [0,c)$\\
\begin{align}
    \EE\bigg[\int_{\mathbb{I}}\exp\bigg(c_1 \int_0^T |X_t^\gamma + b_t^{\gamma}|^2\dd t\bigg) \nu(\dd \gamma)\bigg]< \infty.
\end{align}
\end{lemma}

\begin{proof}
Fix a $c_1 \in [0,c)$. Then, by the triangle inequality
\begin{align}
    \exp \bigg(c_1 \int_0^T | X_t^{\gamma} + b_t^{\gamma} |^2 \dd t\bigg) \leq&\: C \exp\bigg(c_1 \int_0^T | X_t^{\gamma} |^2  + c_0|X_t^{\gamma}| \dd t\bigg)\\
    =& \: C\exp\bigg(c_1 \int_0^T (| X_t^{\gamma} |^2  + c_0|X_t^{\gamma}|) \bone_{\{|X_t^{\gamma}| \leq 
    c_0c_1/(c - c_1)\}} \dd t\bigg) \\
    &\times \exp\bigg(c_1 \int_0^T (| X_t^{\gamma} |^2  + c_0|X_t^{\gamma}|) \bone_{\{|X_t^{\gamma}| > c_0c_1/ (c - c_1)\}} \dd t\bigg) \\
    \leq& \: C_1 e^{c\int_0^T |X_t^{\gamma}|^2 \dd t},
\end{align}
where $c_0, C, C_1 >0$ are constants. Integrating with respect to $\nu(\cdot)$ and taking the expectations on the both sides of the inequality concludes the proof.
\end{proof}

\begin{proof}[Proof of Theorem \ref{thm:true.fpp}]
Without loss of generality let $\gamma_0 = \sup\mathbb{I}$. Fix some $T>0$. Let us first show that $U_t(X_t^{\pi})$ is a local FPP. For that it is enough to check that the condition \eqref{eq:fubini.suff} in Theorem \ref{thm:local.FPP} is satisfied. Using H\"{o}lder's inequality we get
\begin{align}
    \EE\bigg[\int_{\mathbb{I}} \bigg(\int_0^t (X_s^{\pi})^{2(1-\gamma)}\big(\ee(M_t^{\gamma})\ee(V_t^{\gamma})\big)^2 & \bigg( \bigg|\frac{1}{1-\gamma}H_s^{\gamma} + \sigma_s\pi_s\bigg|^2 + \bigg|\frac{1}{1-\gamma}J_s^{\gamma}\bigg|^2 \bigg) \dd s\bigg)^{\frac{1}{2}}\nu(\dd \gamma)\bigg] \\
    < & \: \nu(\mathbb{I})^{\frac{1}{2}}\times \EE\bigg[\int_{\mathbb{I}} \int_0^t (X_s^{\pi})^{2(1-\gamma)}\big(\ee(M_t^{\gamma})\ee(V_t^{\gamma})\big)^2\\
    & \times \! \bigg( \bigg|\frac{1}{1-\gamma}H_s^{\gamma} + \sigma_s\pi_s\bigg|^2 \! + \bigg|\frac{1}{1-\gamma}J_s^{\gamma}\bigg|^2 \bigg) \nu(\dd \gamma) \dd s \bigg]^{\frac{1}{2}}\!.
\end{align}
Proving that the above expectation is finite will automatically yield that condition \eqref{eq:fubini.suff} holds.
Before that, let us first consider the process
\begin{align}\label{eq:thm.fpp.proof.u.tilde}
\begin{split}
    \tilde{U}_t(X_t^{\pi}) :=&\: \int_0^t \int_{\mathbb{I}} (X_s^{\pi})^{1-\gamma}\ee(M_s^{\gamma})\ee(V_s^{\gamma}) D^{\gamma}(\pi_s) \nu (\dd \gamma) \: \dd s  \\
    & + \sum_{i = 1}^{d_W}  \int_0^t \int_{\mathbb{I}} (X_s^{\pi})^{1-\gamma}\ee(M_s^{\gamma})\ee(V_s^{\gamma}) \\
    &\times \bigg(\frac{\gamma}{(1-\gamma)\gamma_0} H_s^{\gamma_0} + \frac{\gamma - \gamma_0}{(1-\gamma)\gamma_0}\lambda_s + \sigma_s \pi_s \bigg)_i \, \nu (\dd \gamma)\, \dd W_{s,\, i}\\
    & + \sum_{j = 1}^{d_{\W}} \int_0^t \int_{\mathbb{I}} \frac{1}{1-\gamma} (X_s^{\pi})^{1-\gamma}\ee(M_s^{\gamma})\ee(V_s^{\gamma}) J_{s,\,i}^{\gamma}\, \nu (\dd \gamma)\, \dd \W_{s,\,j} \\
    =:& \:\mathrm{Drift} + (W\textrm{-local mtg.}) + (\W\textrm{-local mtg.}).
\end{split}
\end{align}
Note that for all $\gamma \in \mathbb{I}$ the drift terms $D^{\gamma}(\cdot)$ are non-positive, and $D^{\gamma}(\pi^*) = 0$. Thus, $\tilde{U}_t(x)$ is a local FPP. To prove that this process is a true FPP, it is enough to show two things. First, local martingale summands are indeed true martingales for all admissible portfolios. Second, the suggested optimal portfolio is indeed admissible.

Before doing all that, let us first show that 
$$\sup_{t\in [0,T]} \EE \bigg[ \int_{\mathbb{I}} \big(\ee (M_t^{\gamma})\ee(V_t^{\gamma})\big)^{\frac{2v}{v-1}}\nu(\dd \gamma) \bigg] < \infty.$$ 
For sake of notational simplicity denote $q := \frac{2v}{v-1} \geq 2$. Taking $p_1, p_2, p_3, p_4 \geq 1$ such that 
$$\frac{1}{p_1} + \frac{1}{p_2} + \frac{1}{p_3} + \frac{1}{p_4} = 1,$$ and applying H\"{o}lder's inequality to the product measure $\pp \times \nu$, we get
\begin{align}
    \EE  \bigg[ &\int_{\mathbb{I}}\big(\ee(M_t^{\gamma})\ee(V_t^{\gamma}) \big)^{q} \nu(\dd \gamma) \bigg]\\
    \leq& \: \EE\bigg[ \int_{\mathbb{I}} \ee\bigg(qp_1 \int_0^t \big(H_s^{\gamma} + (1-\gamma)\lambda_s \big) \cdot \dd W_s + qp_1 \int_0^t J_s^{\gamma} \cdot \dd \W_s\bigg) \nu(\dd \gamma)\bigg]^{\frac{1}{p_1}}\\
    &\times \EE \bigg[ \int_{\mathbb{I}} \exp\bigg(qp_2(\gamma-1)\int_0^t \lambda_s \cdot \dd W_s + \frac{qp_2(\gamma - 1)}{2}\int_0^t |\lambda_s|^2 \dd s \bigg) \nu(\dd \gamma) \bigg]^{\frac{1}{p_2}}\\
    &\times \EE \bigg[ \int_{\mathbb{I}} \exp\bigg(  \frac{qp_3}{2\gamma}\big(q p_1 \gamma - 1  \big)\int_0^t|H_s^{\gamma} + (1-\gamma)\lambda_s|^2 \dd s \bigg) \nu(\dd \gamma) \bigg]^{\frac{1}{p_3}}\\
    &\times \EE \bigg[ \int_{\mathbb{I}}\exp\bigg( \frac{qp_4}{2}\big( qp_1 - 1 \big)\int_0^t |J_s^{\gamma}|^2 \dd s \bigg) \nu(\dd \gamma) \bigg]^{\frac{1}{p_4}} <  \infty,
\end{align}
due to Lemma \ref{lem:exp.integrability}, and integrability conditions \eqref{eq:true.fpp.integrability.J} and \eqref{eq:true.fpp.integrability.H}. The upper bound on the right hand side of the inequality is an increasing function with time, hence 
$$\sup_{t\in[0,T]}\EE \bigg[\int_{\mathbb{I}}\big(\ee(M_t^{\gamma})\ee(V_t^{\gamma}) \big)^{q} \nu(\dd \gamma) \bigg] < \infty.$$

Let us now show that the $W$ and $\W$ local martingale summands in \eqref{eq:thm.fpp.proof.u.tilde} are indeed true martingales. Let us start with the stochastic integral terms with respect to $W$. It is enough to show that the quadratic variation of each summand $i = 1,\ldots, d_W,$ has a finite expectation. Applying H\"{o}lder's inequality for the measure $\nu(\cdot)$, Fubini's theorem and triangle inequality we get
\begin{align}
    \EE & \bigg[\int_0^t \bigg| \int_{\mathbb{I}} (X_s^{\pi})^{1-\gamma}\ee(M_s^{\gamma})\ee(V_s^{\gamma}) \bigg(\frac{\gamma}{(1-\gamma)\gamma_0} H_s^{\gamma_0} + \frac{\gamma - \gamma_0}{(1-\gamma)\gamma_0}\lambda_s + \sigma_s \pi_s \bigg)_i \nu (\dd \gamma)\bigg|^2 \dd s\bigg]\\
    &\leq \begin{aligned}[t] T \times \nu(\mathbb{I})\times &\bigg(\EE \int_0^t  \int_{\mathbb{I}} \frac{\gamma^2}{(1-\gamma)^2 \gamma_0^2}\: (X_s^{\pi})^{2-2\gamma}\big(\ee(M_s^{\gamma}) \ee(V_s^{\gamma})\big)^2 \big|H_{s,\,i}^{\gamma_0}\big|^2  \nu(\dd \gamma) \dd s \\
    &+ \EE \int_0^t  \int_{\mathbb{I}} \frac{(\gamma - \gamma_0)^2}{(1-\gamma)^2 \gamma_0^2}\: (X_s^{\pi})^{2-2\gamma}\big(\ee(M_s^{\gamma}) \ee(V_s^{\gamma})\big)^2 \big|\lambda_{s,\,i}\big|^2 \nu(\dd \gamma) \dd s \\
    &+ \EE \int_0^t  \int_{\mathbb{I}} \: (X_s^{\pi})^{2-2\gamma}\big(\ee(M_s^{\gamma}) \ee(V_s^{\gamma})\big)^2 \big|\sigma_{s,\,i\cdot}\pi_s\big|^2\nu(\dd \gamma) \dd s\bigg)
    \end{aligned}\\
    &=: C \times \big((I) + (II) + (III)\big).
\end{align}
Applying H\"{o}lder's inequality with respect to the measure $\pp \times \nu \times \dd t $ to the terms $(I)$, $(II)$, and $(III)$ respectively we get the following system of inequalities
\begin{align}
    (I)\:&\: \begin{aligned}[t]
    \leq&\:  \bigg(\EE \int_0^t \int_{\mathbb{I}}  (X_s^{\pi})^{2uv(1 - \gamma)}  \nu (\dd \gamma) \dd s \bigg)^{\frac{1}{uv}}\times \bigg(\EE\int_0^t\int_{\mathbb{I}}\big(\ee(M_s^{\gamma})\ee(V_s^{\gamma})\big)^{\frac{2v}{v-1}} \nu(\dd \gamma) \dd s \bigg)^{\frac{v-1}{v}} \\
    &\times \bigg(\EE \int_0^t |H_{s,i}^{\gamma_0}|^{\frac{2uv}{u-1}} \int_{\mathbb{I}} \bigg(\frac{\gamma}{(1-\gamma)\gamma_0}\bigg)^{\frac{2uv}{u-1}}  \nu(\dd \gamma) \dd s  \bigg)^{\frac{u-1}{uv}} < \infty ,
    \end{aligned}\\
    (II) \:&\: \begin{aligned}[t]
    \leq& \:  \bigg(\EE \int_0^t \int_{\mathbb{I}}  (X_s^{\pi})^{2uv(1 - \gamma)}  \nu (\dd \gamma) \dd s \bigg)^{\frac{1}{uv}}\times \bigg(\EE\int_0^t\int_{\mathbb{I}}\big(\ee(M_s^{\gamma})\ee(V_s^{\gamma})\big)^{\frac{2v}{v-1}} \nu(\dd \gamma) \dd s \bigg)^{\frac{v-1}{v}} \\
    &\times \bigg(\EE \int_0^t |\lambda_{s,i}|^{\frac{2uv}{u-1}} \int_{\mathbb{I}} \bigg(\frac{\gamma - \gamma_0}{(1-\gamma)\gamma_0}\bigg)^{\frac{2uv}{u-1}}  \nu(\dd \gamma) \dd s  \bigg)^{\frac{u-1}{uv}} < \infty ,
    \end{aligned}\\
    (III) \:&\: \begin{aligned}[t]
    \leq& \: \bigg( \EE \int_0^t \int_{\mathbb{I}}(X_s^{\pi})^{2v(1-\gamma)}|\sigma_{s,i}\pi_s|^{2v} \nu(\dd \gamma) \dd s\bigg)^{\frac{1}{v}} \\
    &\times \bigg(\EE\int_0^t \int_{\mathbb{I}} \big(\ee(M_s^{\gamma})\ee(V_s^{\gamma})\big)^{\frac{2v}{v-1}} \nu(\dd \gamma) \dd s \bigg)^{\frac{v-1}{v}} < \infty.
    \end{aligned}
\end{align}
This shows that all the stochastic integrals in question are in fact square integrable martingales. 

Now, let us prove that the stochastic integrals with respect to $\W$ are true martingales as well. Similarly, let's consider each summand $j = 1,\ldots,d_{\W}$, and prove that their respective quadratic variations have finite expectations. Proceeding in the exact same way we get
\begin{align}
    \EE \int_0^t & \bigg| \int_{\mathbb{I}} \frac{1}{1-\gamma} (X_s^{\pi})^{1-\gamma}\ee(M_s^{\gamma})\ee(V_s^{\gamma}) J_{s,\,i}^{\gamma}\, \nu (\dd \gamma) \bigg|^2\, \dd s \\
    & \leq
    \begin{aligned}[t]
    &T \times \nu(\mathbb{I})\times \EE \int_0^t  \int_{\mathbb{I}} \frac{1}{(1-\gamma)^2}\: (X_s^{\pi})^{2-2\gamma}\big(\ee(M_s^{\gamma}) \ee(V_s^{\gamma})\big)^2 \big|J_{s,\,i}^{\gamma}\big|^2  \nu(\dd \gamma) \dd s
    \end{aligned}\\
    &\leq
    \begin{aligned}[t]
        &  \bigg(\EE \int_0^t \int_{\mathbb{I}}  (X_s^{\pi})^{2uv(1 - \gamma)}  \nu (\dd \gamma) \dd s \bigg)^{\frac{1}{uv}}\times \bigg(\EE\int_0^t\int_{\mathbb{I}}\big(\ee(M_s^{\gamma})\ee(V_s^{\gamma})\big)^{\frac{2v}{v-1}} \nu(\dd \gamma) \dd s \bigg)^{\frac{v-1}{v}} \\
    &\times \bigg(\EE \int_0^t \int_{\mathbb{I}} \: \bigg|\frac{J_{s,i}^{\gamma}}{1-\gamma}\bigg|^{\frac{2uv}{u-1}} \nu(\dd \gamma) \dd s  \bigg)^{\frac{u-1}{uv}} < \infty. 
    \end{aligned}
\end{align}
This, along with $(I), (II), (III) < \infty$, shows that
\[\EE\bigg[\int_{\mathbb{I}} \int_0^t (X_s^{\pi})^{2(1-\gamma)}(\ee(M_s^{\gamma})\ee(V_s^{\gamma}))^2\bigg( \bigg|\frac{1}{1-\gamma}H_s^{\gamma} + \sigma_s\pi_s\bigg|^2 + \bigg|\frac{1}{1-\gamma}J_s^{\gamma}\bigg|^2 \bigg) \nu(\dd \gamma) \:\dd s \bigg]^{\frac{1}{2}}<\infty,\]
and thus condition \eqref{eq:fubini.suff} holds. Applying the stochastic Fubini Theorem (\cite[Theorem 2.2]{veraar2012stochastic}) yields that $U_t(X_t^{\pi}) = \tilde{U_t}(X_t^{\pi})$ for all $\pi \in \A^{v}$.
\bigskip

Last but not least, let us show that $\pi^*$ is indeed admissible, that is $\pi^*$ satisfies conditions \eqref{eq:admiss.integral} and \eqref{eq:admiss.sup}. Applying H\"{o}lder's and then Minkowski's inequalities we get
\begin{align}
    \EE \int_0^T \int_{\mathbb{I}} \: (X_t^{\pi^*})^{2v(1-\gamma)} |\sigma_t\pi_t^*|^{2v} \nu(\dd \gamma) \dd t =& \: C_0 \E \int_0^T \int_{\mathbb{I}} \: (X_t^{\pi^*})^{2v(1 - \gamma)} |\lambda_t + H_t^{\gamma_0} |^{2v} \nu(\dd \gamma)\dd t \\
   \leq& \: C_0 \bigg( \E \int_0^T \int_{\mathbb{I}} \: (X_t^{\pi^*})^{2uv(1-\gamma)} \nu(\dd \gamma )\dd t\bigg)^{\frac{1}{u}} \\
   &\times \Bigg(\bigg(  \E \int_0^T\int_{\mathbb{I}} \: | \lambda_t |^{\frac{2uv}{u-1}}\nu(\dd \gamma )\dd t\bigg)^{\frac{u-1}{2uv}} 
   \\ &+ \bigg(  \E \int_0^T \int_{\mathbb{I}} \:| H_t^{\gamma_0} |^{\frac{2uv}{u-1}}\nu(\dd \gamma)\dd t\bigg)^{\frac{u-1}{2uv}} \Bigg)^{2v} \\
   \leq& \: C_1 \bigg( \E \int_0^T \int_{\mathbb{I}} \:  (X_t^{\pi^*})^{2uv(1-\gamma)} \nu(\dd \gamma) \dd t\bigg)^{\frac{1}{u}},
\end{align}
since $\lambda$ is bounded and the supremum of the corresponding moments of $H_t^{\gamma_0}$ are bounded. Thus to verify \eqref{eq:admiss.integral} it is enough to verify \eqref{eq:admiss.sup}. For the convenience of the reader, let us write out the expression for the optimal wealth component
\begin{align}
    X_T^{\pi^*} &=
    \begin{aligned}[t]
    X_0 \exp\bigg( & \frac{1}{\gamma_0} \int_0^T \big(H_t^{\gamma_0} + (1-\gamma_0)\lambda_t\big)\cdot \dd W_t - \frac{1}{2\gamma_0^2} \int_0^T \big| H_t^{\gamma_0} + (1-\gamma_0)\lambda_t \big|^2 \dd t  \\
    &+ \int_0^T \lambda_t \cdot \dd W_t + \frac{1}{2} \int_0^T |\lambda_t|^2 \dd t \bigg). 
    \end{aligned}
\end{align}
Once again, taking the selection $p_1,p_2,\tilde{p}_3 >1$ such that $\frac{1}{p_1} + \frac{1}{p_2} + \frac{1}{\tilde{p}_3} = 1$, denoting $uv =: w$, and applying H\"{o}lder's inequality to the measure $\pp \times \nu$ yields
\begin{align}
    \EE& \bigg[\int_{\mathbb{I}} (X_T^{\pi^*})^{2w(1-\gamma)} \nu(\dd \gamma)\bigg]\\
    &\leq 
    \begin{aligned}[t]
    &\EE\bigg[\int_{\mathbb{I}} X_0^{2wp_1(1-\gamma)}\ee\bigg(\frac{2wp_1(1-\gamma)}{\gamma_0}\int_0^T \big(H_t^{\gamma_0} + (1 - \gamma_0)\lambda_t \big) \cdot \dd W_t\bigg) \nu (\dd \gamma)  \bigg]^{\frac{1}{p_1}}\\
    &\times \EE\bigg[ \int_{\mathbb{I}} \exp\bigg(\frac{wp_2(1-\gamma)}{\gamma_0^2}(2wp_1(1-\gamma) - 1)\int_0^T\big|H_t^{\gamma_0} + (1- \gamma_0)\lambda_t\big|^2 \dd t\bigg) \nu (\dd \gamma)\bigg]^{\frac{1}{p_2}} \\
    &\times \EE\bigg[\int_{\mathbb{I}} \exp\bigg( 2w\tilde{p}_3(1-\gamma)\int_0^T \lambda_t \cdot \dd W_t + w\tilde{p}_3(1-\gamma) \int_0^T |\lambda_t|^2  \bigg) \nu (\dd \gamma)\bigg]^{\frac{1}{\tilde{p}_3}}.
    \end{aligned}
\end{align}
Since $\lambda_t$ is bounded and $\mathbb{I}$ is compact we get
\begin{align}
    \EE \bigg[\int_{\mathbb{I}} (X_T^{\pi^*})^{2w(1-\gamma)} \nu(\dd \gamma)\bigg] &\leq
    \begin{aligned}[t]
    &C \bigg(\int_{\mathbb{I}} X_0^{2wp_1(1-\gamma)}  \nu(\dd \gamma) \bigg)^{\frac{1}{p_1}} \times \EE\bigg[ \int_{\mathbb{I}} \exp\bigg(\frac{wp_2(1-\gamma)}{\gamma_0^2}\\
    &\times(2wp_1(1-\gamma) - 1) \int_0^T\big|H_t^{\gamma_0} + (1- \gamma_0)\lambda_t\big|^2 \dd t\bigg) \nu (\dd \gamma)\bigg]^{\frac{1}{p_2}}
    \end{aligned}\\
    & < \infty,
\end{align}
due to Lemma \ref{lem:exp.integrability} and condition \eqref{eq:true.fpp.integrability.H}. Since the obtained bounds are increasing functions in $T$, we get that $\pi^*$ does indeed satisfy \eqref{eq:admiss.sup}. Thus $\pi^* \in \A^v$, which shows that $\tilde{U}_t(x)$, and therefore $U_t(x)$, is a true FPP.
\end{proof}

\section{Proof of Proposition \ref{prop:mixture_components}}\label{appendix.b}
\begin{proof}[Proof of Proposition \ref{prop:mixture_components}]
Just like previously, since $A_t, D_t > 0$ are continuous semimartingales, they can be represented in the following form
\begin{align}\label{eq:two-power_A_dynamics}
    \dd A_t & = \alpha_t A_t \dd t + a_t A_t \cdot \dd W_t + a_t^{\perp}A_t \cdot \dd \W_t,\\\label{eq:two-power_D_dynamics}
    \dd D_t & = \delta_t D_t \dd t + d_t D_t \cdot \dd W_t + d_t^{\perp}D_t \cdot \dd \W_t.
\end{align}
Taking any admissible portfolio $\pi$ and applying It\^{o}'s formula to $U_t(X_t^{\pi})$ yields
\begin{align}
    \dd U_t(X_t^{\pi}) =& \:\Big[ \alpha_t A_t (X_t^{\pi})^p \! + \! \delta_t D_t (X_t^{\pi})^q \! + \!  (\sigma_t \pi_t)^{\top} \big( pA_t (X_t^{\pi})^p (\lambda_t + a_t) \! + \! qD_t (X_t^{\pi})^q (\lambda_t + d_t) \big) \\
    &+ \frac{1}{2}|\sigma_t \pi_t|^2\big((p^2 - p)A_t (X_t^{\pi})^p + (q^2 - q)D_t (X_t^{\pi})^q \big) \Big]\dd t\\
    &+ \big((p\sigma_t \pi_t + a_t)A_t(X_t^{\pi})^p + (q\sigma_t \pi_t + d_t)D_t (X_t^{\pi})^q \big)\cdot\dd W_t \\
    &+  \big((X_t^{\pi})^p a_t^{\perp} A_t + (X_t^{\pi})^q d_t^{\perp} D_t\big) \cdot \dd \W_t.
\end{align}
Since $U_{\cdot}(\cdot)$ is a forward performance process, then for all admissible $\pi$ it is necessary that the finite variation term is non-increasing in time
\begin{align}
     \alpha_t A_t (X_t^{\pi})^p &+ \delta_t D_t (X_t^{\pi})^q +  (\sigma_t \pi_t)^{\top} \big( pA_t (X_t^{\pi})^p (\lambda_t + a_t) + qD_t (X_t^{\pi})^q (\lambda_t + d_t) \big) \\
    &+ \frac{1}{2}|\sigma_t \pi_t|^2\big((p^2 - p)A_t (X_t^{\pi})^p + (q^2 - q)D_t (X_t^{\pi})^q \big) \leq 0,
\end{align}
and is equal to $0$ for some admissible $\pi^*$. Note that the term on the left hand side of the above inequality is strictly concave in $\sigma_t \pi_t$. Thus, an optimal $\pi^*$ must be given as a solution to
\begin{align}\label{eq:mixture_portfolio}
    \sigma_t\pi_t^* &= \frac{pA_t (X_t^{\pi^*})^p (\lambda_t + a_t) + qD_t (X_t^{\pi^*})^q (\lambda_t + d_t)}{p(1-p)A_t (X_t^{\pi^*})^p + q(1-q)D_t (X_t^{\pi^*})^q},
\end{align}
where $X_t^{\pi^*}$ is the corresponding optimal wealth process. Plugging it in back into the drift integrand must equalize it to 0 
\begin{align}
\frac{1}{2}\frac{|pA_t (X_t^{\pi^*})^p (\lambda_t + a_t) + qD_t (X_t^{\pi^*})^q (\lambda_t + d_t)|^2}{p(1-p)A_t (X_t^{\pi^*})^p + q(1-q)D_t (X_t^{\pi^*})^q} + \alpha_t A_t (X_t^{\pi^*})^p + \delta_t D_t (X_t^{\pi^*})^q = 0.
\end{align}
Bringing everything to a common denominator the equation becomes
\begin{align}
    p^2|\lambda_t &+ a_t|^2 A_t^2(X_t^{\pi^*})^{2p} + 2pq(\lambda_t + a_t)^{\top}(\lambda_t + d_t)A_tD_t(X_t^{*})^{p+q} + q^2|\lambda_t + d_t|^2 D_t^2 (X_t^{\pi^*})^{2q} \\ =&\:-\big(2p(1-p)A_t^2 \alpha_t (X_t^{\pi^*})^{2p} +  2\big(p(1-p)\delta_t + q(1-q) \alpha_t\big) A_tD_t(X_t^{\pi^*})^{p+q} \\
    &+  2q(1-q)D_t^2 \delta_t (X_t^{\pi^*})^{2q} \big).
\end{align}
Matching the coefficients in front of the $2p$ and $2q$ powers of $X_t^{\pi^*}$ we get the following necessary characterization for $\alpha_t$ and $\delta_t$
\begin{align}\label{eq:two-power_fpp_FV}
    &\alpha_t = -\frac{p^2}{2p(1-p)}|\lambda_t + a_t|^2,\quad \delta_t = -\frac{q^2}{2q(1-q)}|\lambda_t + d_t|^2.
\end{align}
Plugging these back into our original equation we obtain the following necessary condition
\begin{align}
    \bigg( \frac{p^2q(1-q)}{p(1-p)} |\lambda_t + a_t|^2 \! + \! \frac{q^2p(1-p)}{q(1-q)}|\lambda_t + d_t|^2 \! - \! 2pq(\lambda_t + a_t)^T(\lambda_t + d_t)\bigg)A_tD_t(X_t^{*})^{p+q} \! &= 0.
\end{align}
Since $A_t, D_t >0$ and $0<p,q<1$, we obtain that
\begin{align}
    \bigg|\frac{1}{1-p}(\lambda_t + a_t) - \frac{1}{1-q}(\lambda_t + d_t) \bigg|^2 = 0.
\end{align}
This yields
\begin{align}
    \frac{1}{1-p}(\lambda_t + a_t) = \frac{1}{1-q}(\lambda_t + d_t).
\end{align}
Note that the derived conditions on $a_t, d_t, \alpha_t, \delta_t$ exactly match the characterization from Theorem \ref{thm:local.FPP}, where we choose $\nu$ to be a measure with two atoms located at $\{1-p, 1-q\}$. Thus, we know that $A_tx^p$ and $D_tx^q$ are forward performance processes with the optimal portfolio given by a solution to
\begin{align}
    \sigma_t\pi_t^* = \frac{1}{1-p}(\lambda_t + a_t).
\end{align}
\end{proof}

\section*{Acknowledgments}
We would like to thank Mykhaylo Shkolnikov for his valuable suggestions and multiple fruitful discussions.

\bibliographystyle{plain}
\bibliography{bibliography_www}

\end{document}